\documentclass[aps,nofootinbib,english,prx,floatfix,amsmath,superscriptaddress,tightenlines,twocolumn]{revtex4-2}

\usepackage[utf8]{inputenc}
\usepackage[T1]{fontenc}
\usepackage{amsmath,amssymb,amsfonts, amsthm}
\usepackage{mathtools}
\usepackage{physics}
\usepackage{graphicx}
\usepackage{xcolor}
\usepackage{hyperref}
\usepackage{cleveref}
\usepackage{booktabs}
\usepackage{array}
\usepackage{dsfont}
\usepackage{subcaption}
\usepackage[font=small]{caption}
\usepackage{enumerate}
\usepackage{tikz}
\usepackage{pgfplots}
\usepackage{bbold}
\pgfplotsset{compat=1.18}
\hypersetup{
    colorlinks=true,
    linkcolor=blue,
    filecolor=magenta,      
    urlcolor=cyan,
    citecolor=red,
}

\newtheorem{theorem}{Theorem}
\newtheorem{lemma}[theorem]{Lemma}

\newtheorem{corollary}[theorem]{Corollary}
\newtheorem{definition}[theorem]{Definition}

\begin{document}

\title{Unlearnable phases of matter}

\author{Tarun Advaith Kumar}
\email{tkumar@perimeterinstitute.ca}
\thanks{These authors contributed equally to this work.}
\affiliation{Perimeter Institute for Theoretical Physics, Waterloo, ON N2L 2Y5, Canada}
\affiliation{Department of Physics and Astronomy, University of Waterloo, Ontario, N2L 3G1, Canada}
\author{Yijian Zou}
\email{yzou@perimeterinstitute.ca}
\thanks{These authors contributed equally to this work.}
\affiliation{Perimeter Institute for Theoretical Physics, Waterloo, ON N2L 2Y5, Canada}
\author{Amir-Reza Negari}
\affiliation{Perimeter Institute for Theoretical Physics, Waterloo, ON N2L 2Y5, Canada}
\affiliation{Department of Physics and Astronomy, University of Waterloo, Ontario, N2L 3G1, Canada}
\author{Roger G. Melko}
\affiliation{Perimeter Institute for Theoretical Physics, Waterloo, ON N2L 2Y5, Canada}
\affiliation{Department of Physics and Astronomy, University of Waterloo, Ontario, N2L 3G1, Canada}
\author{Timothy H. Hsieh}
\email{thsieh@pitp.ca}
\affiliation{Perimeter Institute for Theoretical Physics, Waterloo, ON N2L 2Y5, Canada}
\affiliation{Department of Physics and Astronomy, University of Waterloo, Ontario, N2L 3G1, Canada}

\date{\today}

\begin{abstract}
We identify fundamental limitations in machine learning by demonstrating that non-trivial mixed-state phases of matter are computationally hard to learn. Focusing on unsupervised learning of distributions, we show that autoregressive neural networks fail to learn global properties of distributions characterized by locally indistinguishable (LI) states.  We demonstrate that conditional mutual information (CMI) is a useful diagnostic for LI: we show that for classical distributions, long-range CMI of a state implies a spatially LI partner.   By introducing a restricted statistical query model, we prove that nontrivial phases with long-range CMI, such as strong-to-weak spontaneous symmetry breaking phases, are hard to learn. We validate our claims by using recurrent, convolutional, and Transformer neural networks to learn the syndrome and physical distributions of toric/surface code under bit flip noise. Our findings suggest hardness of learning as a diagnostic tool for detecting mixed-state phases and transitions and error-correction thresholds, and they suggest CMI and more generally ``non-local Gibbsness'' as metrics for how hard a distribution is to learn.
\end{abstract}

\maketitle



\emph{Introduction.--} As machine learning plays an increasingly central role in general, it is important to understand its limitations.  In supervised learning, the task is to learn a function from data to labels, and a notoriously challenging instance is the function which computes the parity of bit strings \cite{Blum94sqdim,kearns98}.  In unsupervised learning, the data is not labeled and the task is to learn properties of the data distribution directly.   Which data distributions are hard to learn and why?   A systematic understanding of which {\it classes} of distributions are hard to learn and their universal properties responsible for unlearnability has been lacking. 

Data is a probability distribution over a typically large domain, i.e., a classical ensemble. In many-body physics, there has been significant recent progress in defining and characterizing phases of matter in ensembles \cite{Hastings_2011, coser2019classification,sang2024mixed,markov,ma2023average,ma2023topological,sohal2024noisy,ellison2025toward,yang2025topologicalmixedstatesphases,zhang2022strange,lessa2025strong,lee2024exact,Lee2025symmetryprotected,wang2023intrinsic,de2022symmetry,vijay2025informationcriticalphasesdecoherence,sang2025mixedstatephaseslocalreversibility,fan2024diagnostics,zou2023channeling,sang2024approximate,zhang2025probingmixedstatephasesquantum,Chen2024unconventional,Lessa_2025mixed, Wang2025anomaly,yi2025universaldecayconditionalmutual,Lu_2023,PhysRevLett.131.200201,lee2022decodingmeasurementpreparedquantumphases, wang2025decoherenceinducedselfdualcriticalitytopological,Sun_2025,sala2024spontaneousstrongsymmetrybreaking,negari2025spacetimemarkovlengthdiagnostic,ma2025circuitbasedcharacterizationfinitetemperaturequantum}, both classical and quantum.  Informally, two states are in the same phase if they are connected by a short-time locally reversible process.  This emerging framework encompasses both equilibrium (local Gibbs) \cite{ma2025circuitbasedcharacterizationfinitetemperaturequantum} and non-equilibrium ensembles \cite{sang2025mixedstatephaseslocalreversibility, coser2019classification, ma2023average}, and especially for the latter, conditional mutual information (CMI) serves as an essential correlation measure for characterizing phases and transitions \cite{markov}. 
These developments offer a novel perspective for thinking about data and in particular ``phases'' of data; in \cite{hu2025localdiffusionmodelsphases, liu2025measurementbasedquantumdiffusionmodels}, the noising and denoising processes of diffusion models was linked to locally reversible processes defining phases, and CMI was used to determine whether denoising could be realized locally or not.  

In this work, we show that non-trivial mixed-state phases (certain classes of distributions) are hard for neural networks to learn.  More precisely, in the unsupervised data-driven learning paradigm whereby a neural network is given samples from a target distribution and asked to generate the same distribution, we show that global information in the non-trivial phase is not learned by the network.  All known non-trivial mixed-state phases are characterized by locally indistinguishable (LI) sets of states which have identical local observables but differ by global observables.  Our claim is that neural networks can learn at most a mixture of such LI states and cannot learn the global observable.  We also find a fundamental relation between LI and CMI for classical distributions: long-range CMI in a state implies the existence of a locally indistinguishable partner, and the reverse implication holds for one-dimensional geometries. 

Technically, we introduce a restricted statistical query learning model that captures the essential features of gradient based data-driven training in autoregressive neural networks. We prove that locally indistinguishable states cannot be learned efficiently within this model, establishing polynomial-time lower bounds that apply to a wide class of physically relevant systems.  On the other hand, we show that one-dimensional distributions with short-ranged CMI can be learned efficiently, consistent with results on Gibbs state learning \cite{chen2025learningquantumgibbsstates}. We then introduce a simple toy model of learning which concretely illustrates the vanishing gradients resulting from LI.  Finally, we provide extensive numerical evidence supporting our theoretical predictions on several architectures (recurrent neural networks (RNN), convolutional neural networks (CNN), and Transformers), and across different target systems, including syndrome distributions and physical distributions of quantum and classical error-correcting codes with noise.



Our work has important implications for both machine learning and physics.  For machine learning, it identifies LI and CMI as important concepts and measures for characterizing hardness of learning, and this could potentially be applied to models trained on real-world data or suggest new approaches for AI safety.  For physics, the hardness of learning serves as a novel computational probe for mixed-state phases and transitions, offering a new diagnostic for both numerical simulations and quantum experiments. As many mixed-state phase transitions are associated with error correction and fault  tolerant thresholds \cite{negari2025spacetimemarkovlengthdiagnostic},
(un)learnability can also be used to detect such error thresholds from both numerics and experiments.

We briefly comment on several related insightful works and their differences with this work. Ref.~\cite{yang_cmi_nns} showed that short-ranged CMI distributions can be efficiently {\it represented} by neural networks; we focus on efficient {\it learnability} and borrow techniques from Ref.~\cite{yang_cmi_nns} for this purpose. Refs.~\cite{Feng_2025,schuster2025hardnessrecognizingphasesmatter} consider the task of identifying the phase of matter of a given state, and they establish hardness by finding sufficiently randomized trivial states that would be recognized as non-trivial; in contrast, our task is data-driven learning of the state/distribution itself.  Refs.~\cite{Barratt_2022,singh2025mixedstatelearnabilitytransitionsmonitored} consider monitored circuit dynamics and the learning task of inferring a global charge in the initial state from the circuit measurements, and they map this learning transition to a mixed state transition.  While the learnability in our context also involves whether global information can be inferred or not, our learning framework is very different, involving neural network training driven directly by data from the target distribution, as opposed to repeated measurements of a state.  
\begin{figure}[t]
\label{fig:diagram_learning}
    \includegraphics[width=\linewidth]{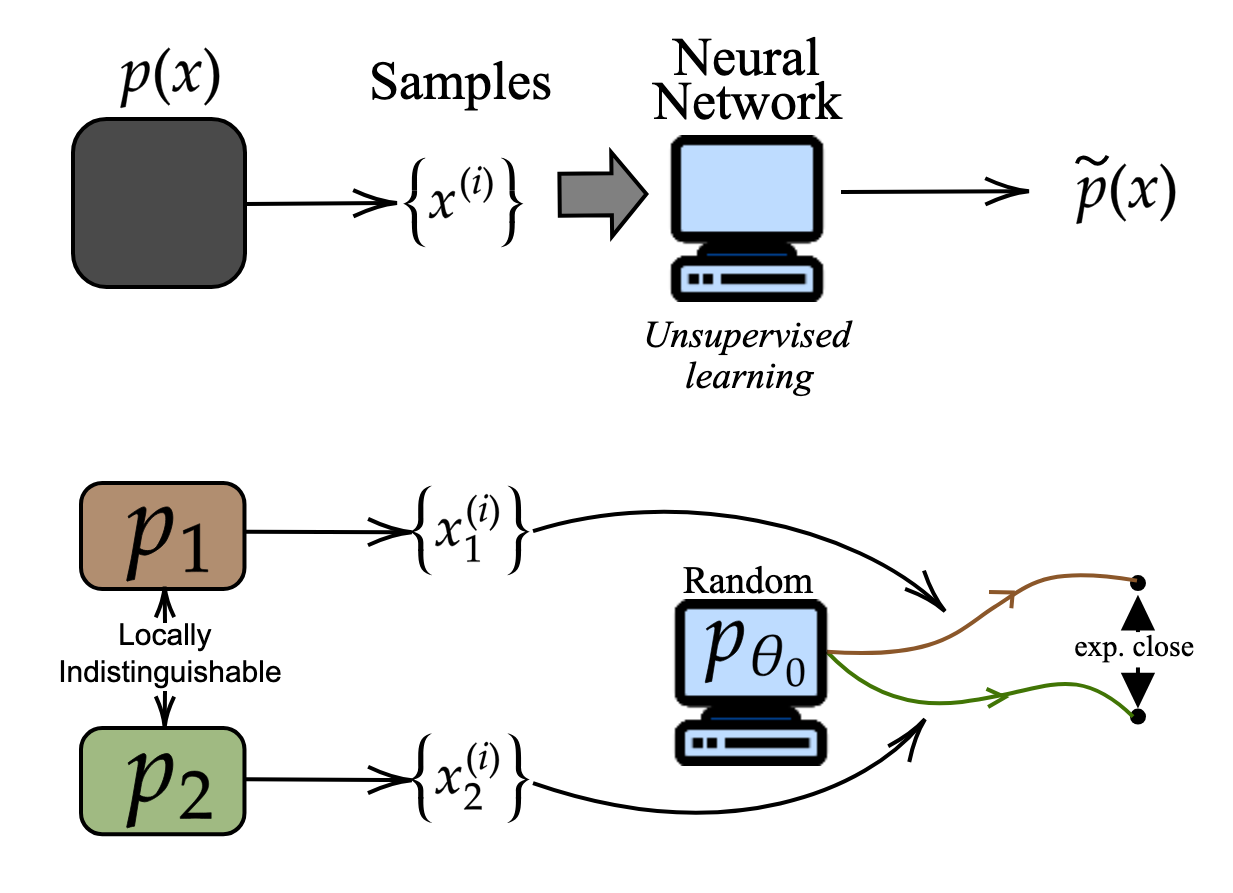}
\caption{\textbf{Unsupervised learning task and local indistinguishability obstruction}. Given samples coming from a black box distribution $p(x)$, our learning problem is to reconstruct the distribution using an autoregressive neural network and unsupervised learning. Training with samples from locally indistinguishable distributions results in exponentially close training trajectories.}
\end{figure}

\emph{Unsupervised learning setup.--} We begin by reviewing the setup for distribution learning, specifically in the context of unsupervised density estimation. We aim to approximate a target distribution $p(x)$ of bitstrings $x \in \{0,1\}^{\otimes n}$ with a variational ansatz $p_{\theta}(x)$ parametrized by $\theta$. The training data is a set of bitstrings, $x^{(1)},x^{(2)}\cdots x^{(M)}$ sampled independently from the distribution $p(x)$. The loss function is given by the negative log-likelihood,
\begin{equation}
\label{eq:loss}
    L_{\theta} = -\frac{1}{M}\sum_{k=1}^M \log p_{\theta}(x^{(k)}).
\end{equation}
The KL divergence $D_{KL}(p||p_{\theta}) = \sum_{x} p(x) \log p(x) -p(x)\log p_{\theta}(x)$, where the second term is approximated by $L_{\theta}$ for large enough $M$. Therefore, minimizing $L_{\theta}$ is equivalent to minimizing the KL divergence.

The training is often achieved through noisy gradient descent (NGD). At time step $t$ we evolve the parameters
\begin{equation}
\label{eq:NGD}
    \theta{(t+1)} = \theta{(t)} - \gamma \nabla_{\theta} L_{\theta{(t)}} + Z_t
\end{equation}
where $\gamma$ is the learning rate, and $Z_t$ is a Gaussian random vector with zero mean and small variance $\sigma^2 = 1/\mathrm{poly}(n)$. 

\emph{Local statistical query model.--} The noisy gradient descent Eq.~\eqref{eq:NGD} for large enough\footnote{The reduction to SQ does not hold for small sample sizes in gradient descent or small batch sizes in stochastic gradient descent. On the contrary, it has been shown that this leads to simulation of any PAC learning, which is stronger than SQ learning \cite{abbe2020polytimeuniversalitylimitationsdeep,abbe2022powerdifferentiablelearningversus}. } sample size $M$ is an instance of a statistical query (SQ) model \cite{abbe2020polytimeuniversalitylimitationsdeep,abbe2022learningreasonneuralnetworks,abbe2022powerdifferentiablelearningversus,Blum94sqdim,diakonikolas2022optimalsqlowerbounds,Blum03SQ,Aslam93SQ,feldman2016statisticalalgorithmslowerbound,reyzin2020statisticalqueriesstatisticalalgorithms} of learning.  In SQ learning, we are given SQ oracles with tolerance $\tau$: an oracle takes an input query function $\phi(x): \{0,1\}^{\otimes n}\rightarrow [-1,1]$ and returns a value $v$ such that $v\in [\mathbb{E}_{x\sim p(x)} (\phi(x))-\tau, \mathbb{E}_{x\sim p(x)} (\phi(x)) +\tau]$. The gradient term in Eq.~\eqref{eq:NGD} is the query result of $\phi_{\theta}(x) = \nabla_{\theta}\log p_{\theta}(x)$ that depends on the variational parameters $\theta(t)$. The tolerance $\tau \sim M^{-1/2}$ comes from the estimation of the expectation value using samples. 

In this work, we define the $\ell$-local statistical query (SQ) model, where we require the query function $\phi(x)$ to depend on at most $\ell$ spins. This locality restriction is central to our framework and can be justified through several complementary perspectives. 

First, from an empirical standpoint, many autoregressive neural networks such as RNNs, CNNs, and transformers exhibit simplicity bias---they prioritize learning local features over global ones. Second, from a physical perspective, when the target distribution is locally indistinguishable (LI) with a local Gibbs state, $e^{-\beta H_{local}}/\Tr[e^{-\beta H_{local}}]$ where $H_{local}$ is a sum of $l$-local operators, we show that noisy gradient descent training (Eq.~\eqref{eq:NGD}) is consistent with local SQ learning (see SM for details). The underlying intuition is inductive. The initial state $\rho_{\theta(0)}$ is a local Gibbs state due to random initialization. Furthermore, the loss function Eq.~\eqref{eq:loss} cannot tell the difference between LI states if $p_{\theta}$ is a local Gibbs state. As a result, the entire training trajectory would be within the space of local Gibbs states and can be simulated by local SQ. 
Third, for certain architectures such as RNNs in the contractive regime, we can explicitly demonstrate that the query function is local (see SM for details).

We note that the notion of locality is architecture-dependent. For RNNs, locality is naturally defined along the quasi-1D autoregressive path, and this is a stronger notion of spatial locality (hereafter referred to as ``$s$-locality''). For deep CNNs and transformers, locality does not necessarily refer to geometric locality but rather to few-body ``$l$-locality'' (or low-weight functions in the Fourier-Walsh expansion of $\phi(x)$).     

Within this framework, we say a distribution $p(x)$ is easy to learn under local SQ if there exists an algorithm that outputs a distribution $\tilde{p}(x)$ using $T = \mathrm{poly}(n)$ local SQ oracle queries with noise level $\tau = 1/\mathrm{poly}(n)$, such that $\|p - \tilde{p}\| = 1/\mathrm{poly}(n)$.

\emph{Learning hardness from local indistinguishability.--} The key physical concept that determines hardness of learning under local SQ is local instinguishablity (LI). Two states $\rho_1$ and $\rho_2$ are LI if (i) for all operators $\phi$ supported on $l_{\phi}$ qubits,
\begin{equation}
\label{eq:LI}
    |\tr((\rho_1 -\rho_2)\phi)| \leq \delta,
\end{equation}
where $\delta \leq O(e^{-n^{1-c_0}/l_{\phi}})$ for a small constant $c_0<1$ 
and (ii) they are globally distinguishable:
\begin{equation}
\label{eq:LI2}
    ||\rho_1-\rho_2||_1 \geq C'
\end{equation}
for some constant $C'$. Again, one can define a stronger notion of spatial locality for LI, called $s$-LI, in which the support of all operators $\phi$ are spatially local.  Intuitively, under $l$-local SQ access with tolerance $\tau = \mathrm{poly}(1/n)$, the oracle responses for $\rho_1$ and $\rho_2$ are statistically indistinguishable, as the expectation values $\tr(\rho_1 \phi)$ and $\tr(\rho_2 \phi)$ differ by at most exponentially small amounts, which is negligible compared to the tolerance $\tau$. 
More formally, we can define a family of LI states $\rho(p_1) = p_1\rho_1 + (1-p_1) \rho_2$, where $p_1$ is a uniformly random variable on $[0,1]$. LI constrains the mutual information between $p_1$ and the oracle result $\mathbf{v}$ to be super-polynomial small:  $I(\mathbf{v}, p_1) \leq T\delta/(4\tau)$.  This implies that at most a superpolynomially small neighborhood of one state in the LI family is easy to learn. 

In the context of training the variational ansatz $p_{\theta}(x)$, we consider the training trajectory $\theta(t)$ under the noisy gradient descent Eq.~\eqref{eq:NGD}. 
Given that the NGD training is described by local SQ, we show that the distributions over training trajectories $\mathbb{P}_1(\theta(t))$ and $\mathbb{P}_2(\theta(t))$ for LI states $\rho_1$ and $\rho_2$, respectively, are superpolynomially close: $||\mathbb{P}_1 - \mathbb{P}_2||_1 \leq \frac{\gamma \delta \sqrt{T}}{2\sigma}$ (see appendix for details of the proof).


\emph{Easiness of learning quasi-1d short-range CMI states. --}
In contrast to the hardness result above, we show that quasi-1d distributions with short-ranged CMI are efficiently learnable under local SQ.

Consider a tripartition of the distribution into $A,B$ and $C$ such that $B$ acts as a buffer region between $A$ and $C$. The decay of the conditional mutual information with the size of $B$ defines the Markov length $\xi$,
\begin{equation}
    I(A:C|B) \leq A \exp(-\frac{|B|}{\xi}),
\end{equation}
and the distribution has short-ranged CMI if $\xi$ is finite.
To establish efficient learnability, we make use of Theorem B.1 in Ref.~\cite{yang_cmi_nns}, which builds the full distribution from its marginals. We partition the system into intervals $A_1, A_2 \cdots A_m$ and use local SQ oracles to perform tomography on each interval $A_i A_{i+1}$ and obtain an approximation of the marginal probability distribution $\tilde{p}_{A_i A_{i+1}} \approx p_{A_iA_{i+1}}$. We choose each interval to have length $l \propto \xi \log n$ and show that each $p_{A_iA_{i+1}}$ can be obtained with $\mathrm{poly}(n)$ local SQ oracles with $\tau = 1/\mathrm{poly}(n)$. Finally, the true distribution $p$ is then approximated by multiplying marginals 
\begin{equation}
\label{eq:MarkovRecon}
    \tilde{p}_{A_1\cdots A_m} = \tilde{p}_{A_1 A_2} \prod_{i=2}^{m-1} \tilde{p}_{A_{i+1}|A_{i}}
\end{equation}
We derive a bound on the error $||p-\tilde{p}||_1 = 1/\mathrm{poly}(n)$ from Theorem B.1 in Ref.~\cite{yang_cmi_nns} and a technical result which bounds error propagation in the Markovian reconstruction Eq.~\eqref{eq:MarkovRecon} of classical distributions (see details in SM). 

\emph{Hardness of learning long-range CMI states.---}
Consider first, for simplicity, a classical distribution $p(x)$ of $n$ spins on a 1D chain. We define a state to have long-range CMI if there exist contiguous intervals $A$, $B$, $C$ such that
   $ I(A:C|B) = O(1)$
where $B$ separates $A$ and $C$ and $|B| = O(n)$. We establish an equivalence between long-range CMI and spatially local indistinguishability ($s$-LI). 

\emph{(a) From long-range CMI to $s$-LI:} If a distribution $p$ has long-range CMI with $|A| = O(1)$, then we can construct a reference distribution $p_{\mathrm{ref}} = p_{AB}p_{C|B}$ that is $s$-LI with $p$ on any contiguous region of size $|B|$. 

\emph{(b) From $s$-LI to long-range CMI:} Conversely, if two distributions $p$ and $q$ are $s$-LI on any contiguous interval of $O(n)$ spins, then at least one of them has long-range CMI (see proofs in SM). 

As a consequence, $s$-local learnability provides a probe for phase transitions between short-range and long-range CMI distributions: short-range CMI distributions are easy to learn under $s$-local SQ, while long-range CMI distributions are hard to learn due to their $s$-LI property.  We emphasize that long-range CMI in a distribution $p(x)$ can arise from $\log p(x)$ having either low-weight (few-body) but spatially non-local components or high-weight components.  In the first case, $s$-local learning is hard, but $l$-local learning may be efficient.  In the second case, we expect that even $l$-local learning is hard, as is the case for the examples studied later. 

The analysis extends to higher-dimensional lattices, though the equivalence is more subtle. Direction (a), long-range CMI $\Rightarrow$ $s$-LI, holds when CMI is computed from a concentric partition where $A$ and $AB$ are concentric balls. Direction (b), $s$-LI $\Rightarrow$ long-range CMI, holds for a quasi-1d partition where $A$ and $B$ are adjacent strips. However, direction (b) \emph{does not hold} for a concentric partition: topologically degenerate states (e.g. ground states of toric code) provide examples that are $s$-LI yet have short-range CMI in a concentric partition. Thus, in higher dimensions, states with long-range CMI on concentric partitions remain hard to learn, but hardness can also arise from topological degeneracy even for states with short-ranged CMI in a concentric partition.

\begin{figure}
    \centering
    \includegraphics[width=\linewidth]{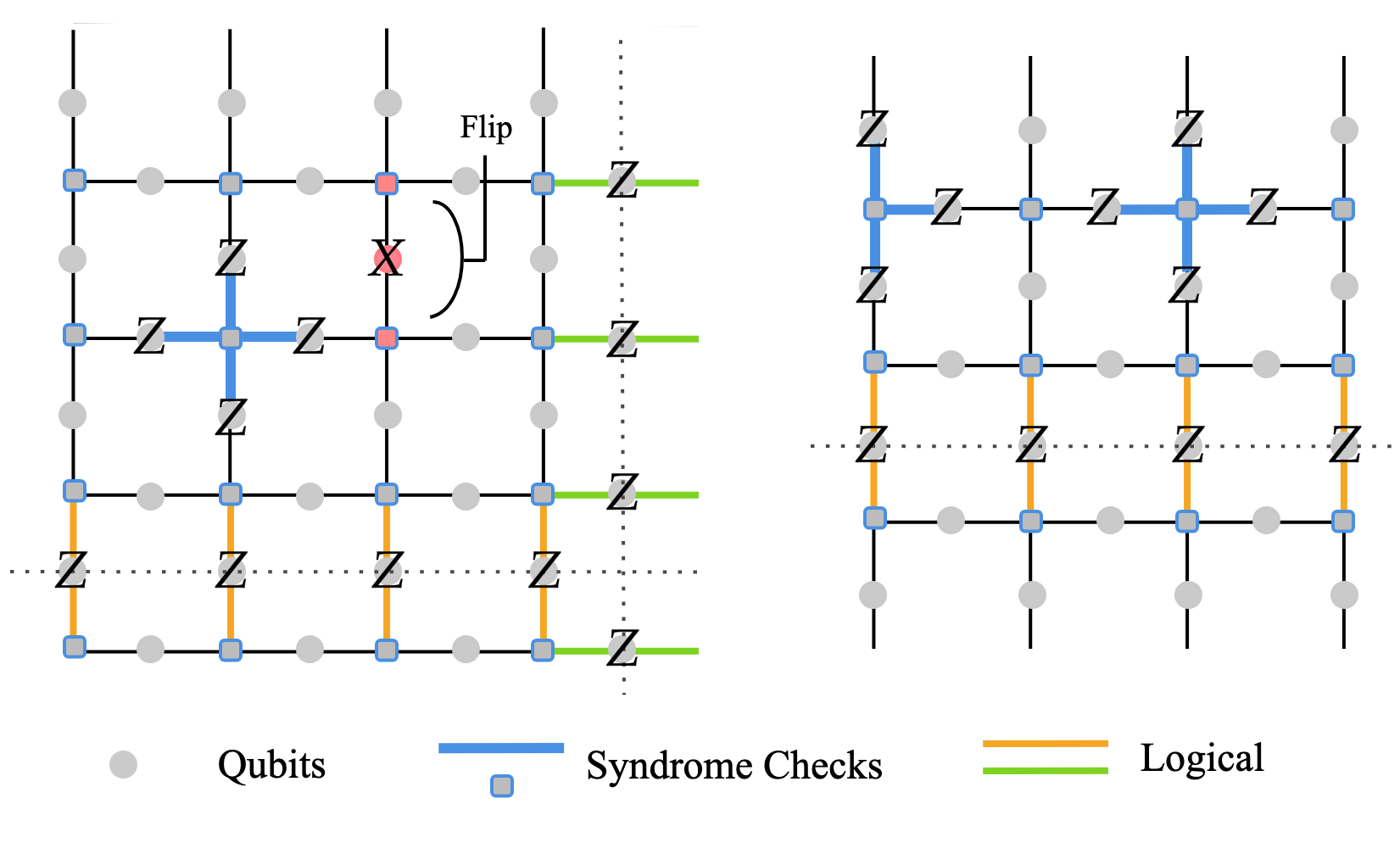}
    \caption{(left) For toric code with bit-flip error on each edge qubit, the relevant syndrome checks and logical operators are shown.  Each bit-flip flips the two syndromes on the nearest vertices. (right) For surface code (open boundary conditions) with bit-flip error, the relevant syndrome checks and logical operators are shown.}
    \label{fig:ToricAndSurfaceCode}
\end{figure}

\emph{Unlearnable phases.--} 
As LI and CMI play an essential role in characterizing mixed-state phases \cite{markov}, the notion of unlearnable states naturally extends to unlearnable phases (families of states with the same qualitative properties). Ref.~\cite{sang2025mixedstatephaseslocalreversibility} defined two mixed states to be in the same phase if they are connected by a low-depth locally reversible channel (in which Markov length remains finite after every operation in the channel).  Moreover, Ref.~\cite{sang2025mixedstatephaseslocalreversibility} showed that the set of LI states with finite Markov length is robust under a locally reversible channel $\mathcal{N}$. The key idea is that the same low-depth channel $\mathcal{D}$ reverses the action of $\mathcal{N}$ on all the LI states, which is sufficient to prove the robustness of LI (see SM for details). Hence, entire phases of LI states (such as topological order) are hard to learn.
Given a global symmetry $P$, a mixed-state can exhibit ``strong-to-weak spontaneous symmetry breaking'' (SWSSB) \cite{lessa2025strong} and have infinite Markov length.  In this case, the set of LI states with different global symmetry charge is also a robust property of the phase.  For example, consider a global Ising symmetry $P=\prod Z$ and a quintessential example of SWSSB, $\sigma_1 = (I + P)/2$, which is the ensemble of all bit strings with even total parity.  $\sigma_1$ is locally indistinguishable from  $\sigma_2 = (I -P)/2$ (bit strings with total odd parity). For any state $\rho_1$ in the same phase as $\sigma_1$, there must be a low-depth strongly-symmetric\footnote{Every gate in the channel must commute with the symmetry $P$.} channel $\mathcal{N}$ such that $\rho_1 \approx \mathcal{N}(\sigma_1)$. Then there exists another state $\rho_2 = \mathcal{N}(\sigma_2)$ that is LI with $\rho_1$, where Eq.~\eqref{eq:LI} follows from the dual channel $\mathcal{N}^{*}$ mapping local operators to local operators and Eq.~\eqref{eq:LI2} follows from the distinct symmetry charges (which are preserved by strongly-symmetric channels). Therefore, the entire strong-to-weak SSB phases are hard to learn by local SQ models.

\emph{Learning a 1d toy model. --} We first illustrate the hardness of learning a nontrivial mixed-state and its relation to error correction in a toy model. Consider a repetition code with codewords $0,..,0$ and $1,..,1$. The syndrome checks are $S_i := Z_i Z_{i+1}$.
A bit-flip error on site $i$ flips the sign of two adjacent syndromes $S_i,S_{i-1}$. Consider a codeword subject to independent bit-flip errors on each site with probability $p$; this results in a syndrome distribution $\rho_p = \sum_{\vec{s}} p(\vec{S}) |\vec{S}\rangle\langle \vec{S}|$. Equivalently, $\rho_p$ can be obtained by applying correlated bit flips  with probability $p$ to each adjacent pair of bits initialized in all $0$ state. Note that since syndromes are flipped in pairs, $p(\vec{S})$ is nonzero only when there are an even number of $-1$ syndromes; the dynamics and states preserve a global Ising symmetry $P=\prod_i S_i=1$.

Optimal decoding involves pairing the defect positions (where $S_i = -1$) and applying physical bit flips along the segments connecting each pair, choosing the minimum-weight perfect matching. As $p$ increases, the typical length of these pairing segments grows.  At the decoding threshold $p=1/2$, logical information is lost because multiple pairings become equally likely, and $\rho_{1/2}= \frac{1+P}{2}$.
As noted in the previous section, this distribution exhibits long-range CMI and SWSSB, and its LI state with short-range CMI is the uniform distribution $\mathbb{1}/2^n$. For $p < 1/2$, the Markov length $\xi$ is finite but diverges as $\xi \sim (1/2 - p)^{-2}$ upon approaching the threshold.

Consider a variational ansatz obtained by starting from $\rho_p$ and applying independent bit-flip noise of rate $r$ to each {\it syndrome} bit, yielding $\rho_{p;r} = \mathcal{N}_r(\rho_p)$. This family corresponds to a restricted class of recurrent neural networks with hidden dimension 4, as $\rho_{p;r}$ is a matrix-product density operator that can be mapped to RNNs via the formalism of Ref.~\cite{Wu_2023_MPSRNN}.

As a toy model of learning, we train this ansatz $\rho_{p;r}$ via gradient descent to approximate the target $\rho_{p}$.  For the target $\rho_{p=1/2}$, the gradient of the loss with respect to $r$ satisfies $\nabla_r L_r \sim (1-2r)^n$ near typical initializations. This gradient vanishes exponentially with system size unless $r=0$ exactly, implying exponentially slow training whenever the optimizer is initialized away from the target—typically near the LI uniform distribution $\mathbb{1}/2^n$. The difficulty persists even for targeting finite but large Markov length $\xi$ below threshold, in which $\nabla_r L_r \sim \xi^{-2}$ at $r\approx 1/2$ (see SM), showing that learning becomes harder as $\xi$ increases. Thus, the onset of unlearnability coincides precisely with the error-correction threshold and the divergence of the Markov length.

\begin{figure}[t]
\begin{subfigure}{0.48 \textwidth}
    \includegraphics[width=\linewidth]{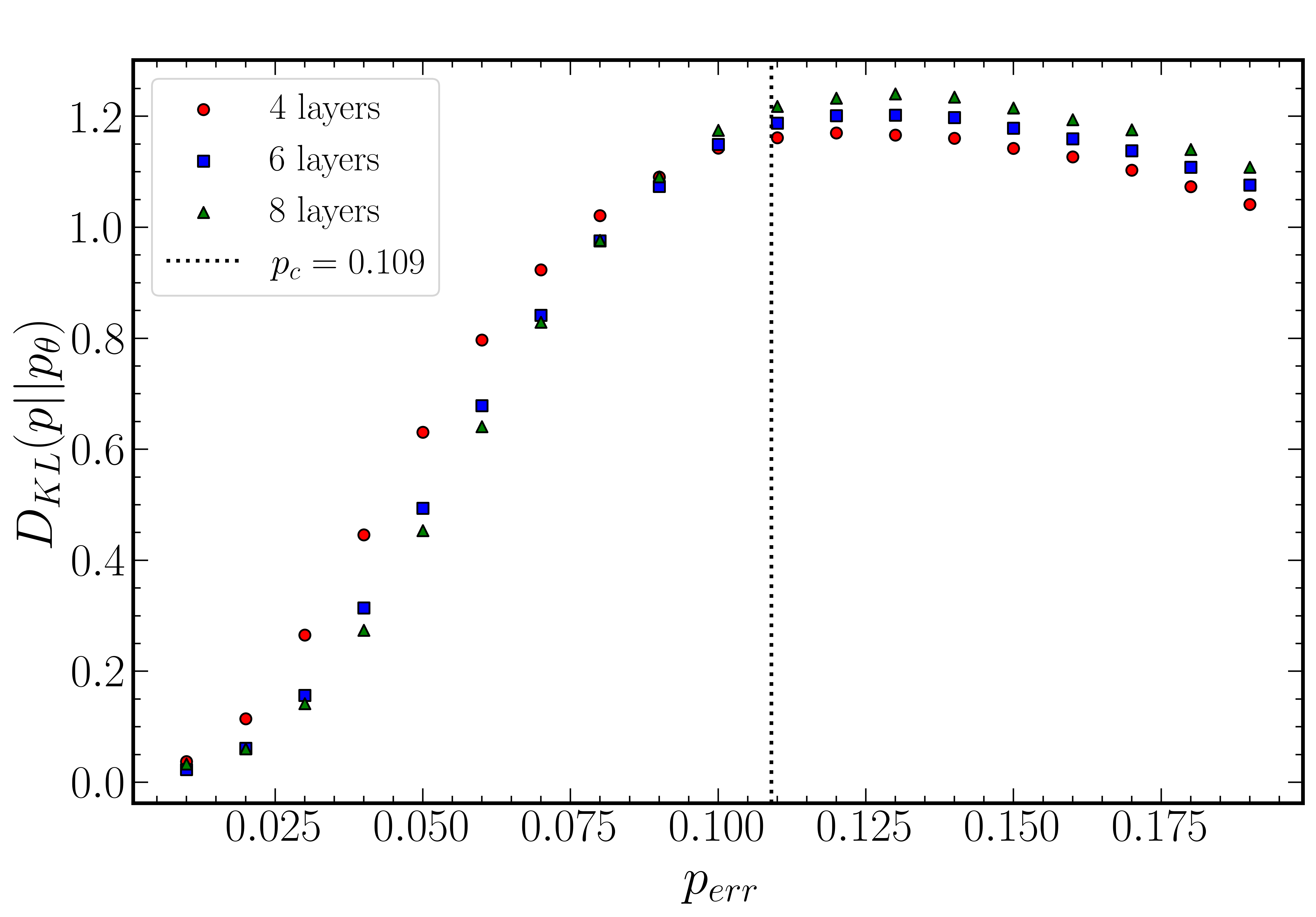}
\end{subfigure}
\hspace{2mm}
\begin{subfigure}{0.48 \textwidth}
    \includegraphics[width=\linewidth]{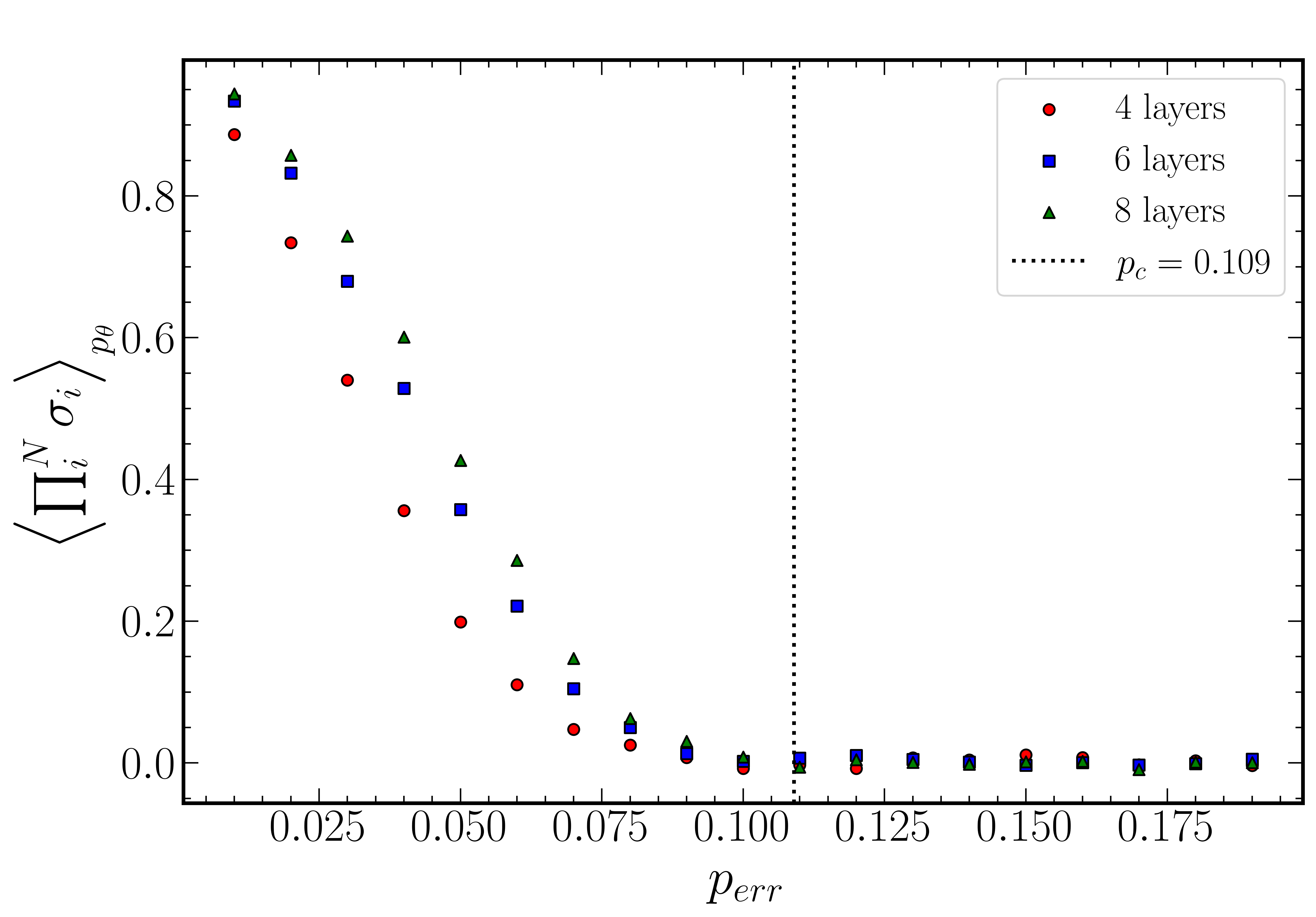}
\end{subfigure}
\caption{\textbf{Learning the syndrome distribution of noisy toric code}. (top) KL divergence between true distribution and trained 2D CNN with residual connections for the toric code vertex syndromes from bit flip error rate $p_{err}$ on 7 by 7 lattice (of syndromes bits). For $p_{err} > 0.109$, the syndrome distribution is in a SWSSB phase and is hard to learn. 
Specifically, (bottom) in the SWSSB phase, the neural network fails to learn the global parity of the true distribution.}
\label{fig:toric_syndrome}
\end{figure}

\textit{Learning the syndrome distribution of the toric code.}-- The toric code is a quantum error correcting code with qubits on the edges of a square lattice and syndrome checks $A_p=\prod_{i\in p} X_i$ and $B_v \equiv \prod_{i\in v}Z_i$, where $p,v$ denote plaquettes and vertices. With periodic boundary conditions, the codespace (in which all the checks are satisfied) is four-dimensional, encoding two logical qubits with logical operators ${\tilde X}_{1,2}, {\tilde Z}_{1,2}$ consisting of products of Pauli $X$ and $Z$ operators, respectively. 
We consider the action of uniform $X$ dephasing, flipping each qubit with probability $p_{err}$
\begin{equation}
\label{eq:bitflip}
    \mathcal{N}^{[i]}_{p_{\text{err}}}(\rho) = (1-p_{err})\rho + p_{err} X_i \rho X_i,
\end{equation}
on a codestate in which all syndromes are initially trivial.
See Fig. 2 (left), which illustrates the $Z$ elements only, as the $X$ elements are not affected by this noise model.

The syndrome distribution $\rho_{\text{syn}}(p_{err})$ is the (classical) distribution of $B_v$ syndromes resulting from noise rate $p_{err}$. Physically, the distribution can be obtained by first initializing the syndrome bits for all vertices to all $0$, and applying with probability $p_{err}$ the $X_{i} X_{j}$ bit flip on each edge $\langle ij\rangle$ (since each physical qubit flip violates the two nearest checks). The total number of syndromes is even for all states in the ensemble, thus $\rho_{\text{syn}}$ has a strong $\mathbb{Z}_2$ symmetry, implying $\langle \prod_i Z_i\rangle = 1$.

As shown in Ref.~\cite{lessa2025strong}, $\rho_{\text{syn}}$ exhibits a phase transition from trivially symmetric to SWSSB, and the transition can be mapped to the 2d random bond Ising model on the Nishimori line. When $p_{\text{err}}>p_c\approx 0.11$, the resulting distribution lies in the SWSSB phase and is two-way connected by a strongly-symmetric finite-depth channel to the parity distribution at $p_{\text{err}}=1/2$. In this phase, we have long-range CMI and there exists a LI partner $\tilde{\rho}_{\text{syn}}$ which is odd under the $\mathbb{Z}_2$ symmetry. Our theory predicts that only the equal mixture of $\rho_{\text{syn}}$ and $\tilde{\rho}_{\text{syn}}$ is efficiently learnable. In particular, the expectation value of the global parity $\langle \prod_i Z_i\rangle$ is not efficiently learnable through local SQ. 

We use a CNN ansatz $p_{\theta}$ to learn the distribution and the result is shown in Fig.~\ref{fig:toric_syndrome}. We have also used RNN and transformer and obtain similar results (see SM).
Both the KL divergence $D_{KL}(p_{\text{syn}}||p_{\theta})$ and the error in $\langle \prod_i Z_i\rangle$ decrease as we increase the CNN depth for a fixed training time if $p<p_c$, while for $p>p_c$ the KL divergence is lower bounded by $\log 2$ and the expectation value $\langle \prod_i Z_i\rangle_{\rho_{\theta}} = 0$ in the final training step. This agrees with our prediction that the global strong symmetry charge is hard to learn in the SWSSB phase. 

Our result provides a way to diagnose both SWSSB phases (and error correction thresholds) through samples of the distribution: the SWSSB phase is characterized by the failure of learning the strong symmetry charge under unsupervised learning with neural networks. Importantly, we do not need access to the value of $p(x)$. This is in contrast with other measures, such as the fidelity correlator and the CMI, which typically require strong simulation of $p(x)$ for efficient computation. 

\begin{figure}[t]
\begin{subfigure}{0.48 \textwidth}
    \includegraphics[width=\linewidth]{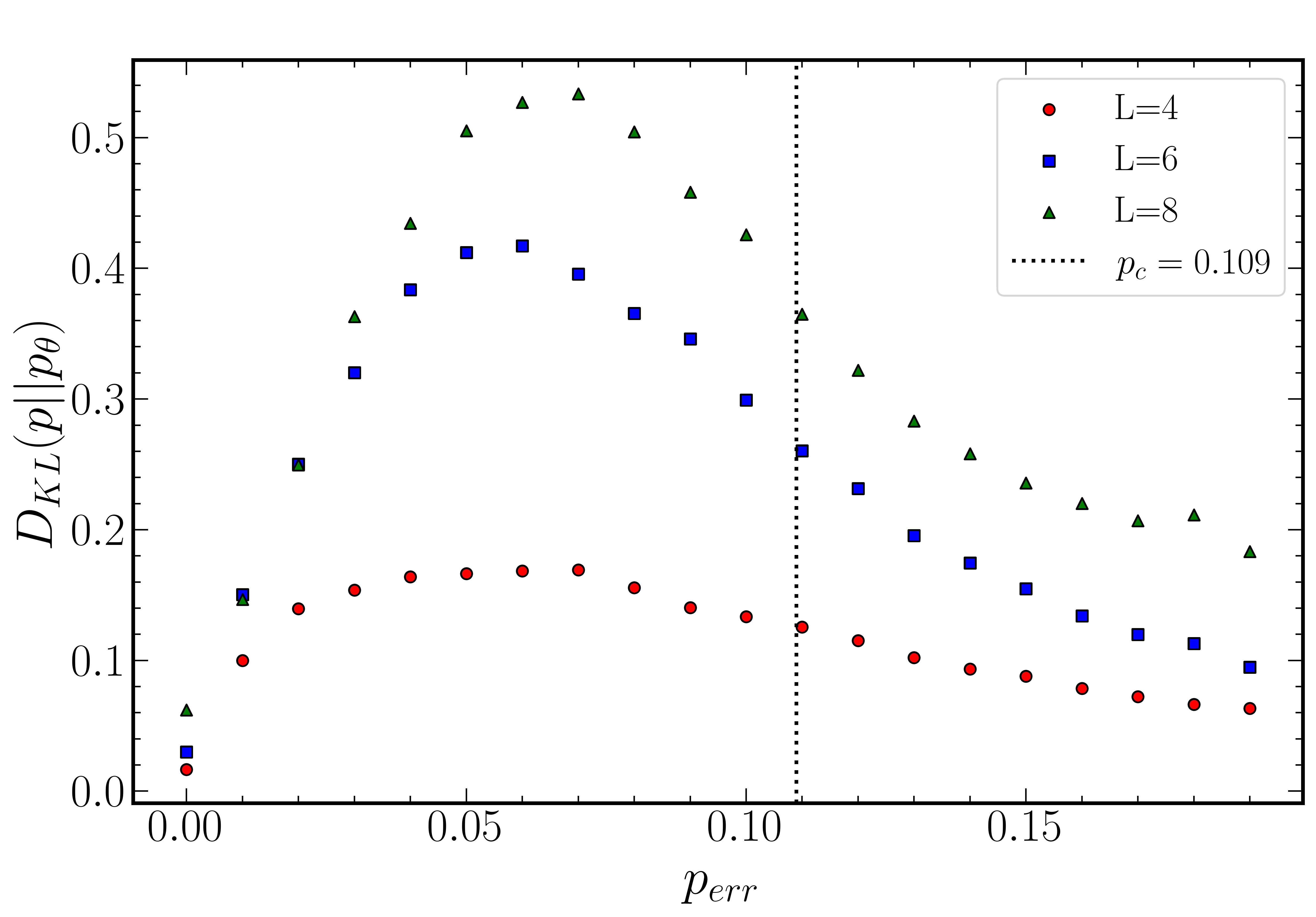}
    \caption{}
    \label{fig:loopModel1}
\end{subfigure}
\hspace{2mm}
\begin{subfigure}{0.48 \textwidth}
    \includegraphics[width=\linewidth]{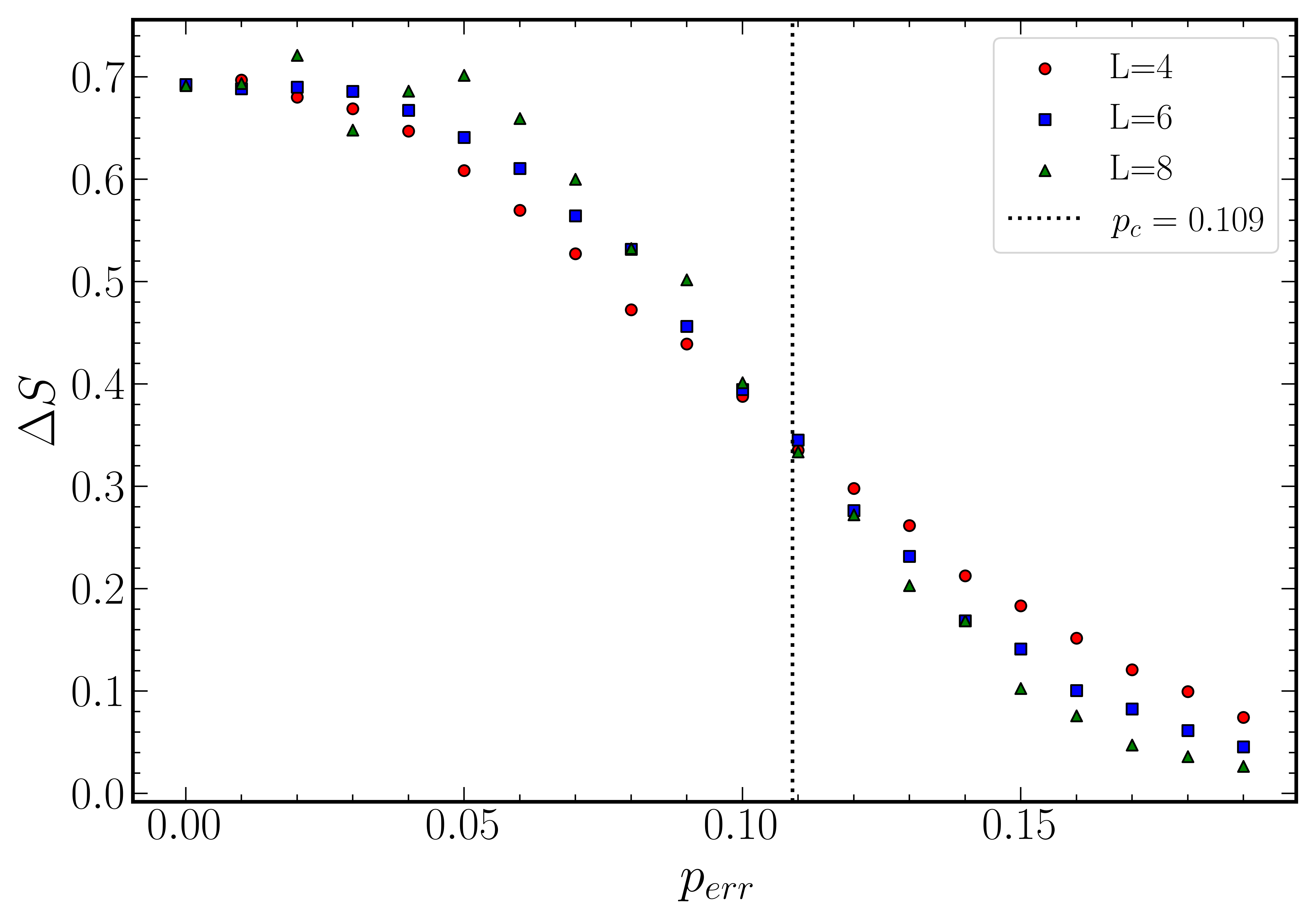}
    \caption{}
    \label{fig:loopModel2}
\end{subfigure}
\caption{\textbf{Learning a noisy loop ensemble}. (top) Error in CNN learning of the classical loop ensemble subject to bit-flip error rate $p_{err}$.  We expect that with larger system size and CNN depth, the location of the peak drifts to the critical point $p_c$ at which CMI has power-law decay. The results are from a residually connected CNN architecture with 6, 8, and 10 layers for $L = 4, 6, 8$ respectively.  (bottom) By subtracting off the KL divergence of learning the loop ensemble with no logical information from that of learning the ensemble with fixed logical, we observe that the additional difficulty of learning a fixed logical sector approaches $\log2$ in the nontrivial (error-correctable) phase. This subfigure uses a vanilla 3-layer CNN.}
\label{fig:loopnoise}
\end{figure}

\emph{Surface code under decoherence.--}
We now analyze the hardness of neural network learning of the distribution arising from the surface code \cite{Fowler_2012} under the same bit-flip error $\mathcal{N}^{[i]}_{p_{\text{err}}}$ on every qubit. 
The surface code is adapted from the toric code and defined with open boundary conditions (syndrome checks are depicted in Fig. 2 (right)).  Again, given the error model, the X checks are irrelevant. The target distribution at $p_{\text{err}}=0$ is a classical ensemble $\rho_{\mathrm{loop}}$, which is the zero-temperature Gibbs state of the Hamiltonian
$
    H_{\text{loop}} = -\sum_{v} B_{v}. 
$
$\rho_{\mathrm{loop}}$ is the uniform mixture over all configurations consisting of closed loops and strings ending on the open boundaries.
The surface code has one logical operator $Z_L$, the product of $Z$ operators that connect the top and bottom boundaries.


The mixed-state $\sigma_p = \mathcal{N}_p(\rho_{\mathrm{loop}})$ undergoes a phase transition to a trivial state at $p_{\mathrm{err}}=p_c\approx 0.11$ (the same threshold as the previous syndrome distribution, which is generated from toric code bit-flip errors).  At this transition, the CMI for a concentric partition decays as a power law with the buffer size \cite{markov,sang2025mixedstatephaseslocalreversibility}, as opposed to exponential decay in both phases.  We expect ensembles with quasi-long-ranged CMI are harder to learn than those with exponential decay of CMI due to the following intuition. 
Near the critical point, $\log p(x)$ has more contribution from high weight operators, as opposed to weight-$O(\log n)$ operators in the two phases. Therefore, learning the distribution requires query functions with longer range, which is harder given fixed hyperparameters during training.
Indeed, when learning these distributions with a CNN, we see a peak in the learning error which sharpens and drifts toward the thermodynamic threshold as the system size and CNN depth are increased (Fig. 4a); note that there are large finite size effects as one is effectively contrasting exponentially and algebraically decaying CMI for a small system.  

One can also fix the logical $Z_L =l = \pm 1$, giving the state $\rho_{l} = \rho_{\mathrm{loop}}(1+l Z_L)/2$.  With noise, $\rho_p = \mathcal{N}_p(\rho_{l})$ undergoes a decodability transition at $p_c$.  
If $p<p_c$, the logical information is decodable and $\sigma_p$ (without logical information) and $\rho_p$ (with logical information) are LI \cite{sang2025mixedstatephaseslocalreversibility}. Our result then indicates that learning $\sigma_p$ and $\rho_p$ as the target state produces the same state $\rho_{\theta}$ under training. As we show in SM, this implies that the difference in hardness between learning logical state and no-logical state  
\begin{equation}
    \Delta S = D_{KL}(\rho_p||\rho_{\theta}) - D_{KL}(\sigma_p||\rho_{\theta}),
\end{equation}
is exactly the amount of decodable logical information, which undergoes a phase transition from $\log 2$ to $0$ at $p=p_c$.  This is confirmed by using a CNN to learn these distributions (Fig. 4b).


We can generalize the example above to any Calderbank–Shor–Steane (CSS) code characterized by $X$-type and $Z$-type stabilizer checks. Given any codeword state $|\psi\rangle = \sum_{\vec{z}} c(\vec{z})|\vec{z}\rangle$ that is an eigenstate of a $Z$-type logical operator, it is exponentially hard in the weight of the logical operator to learn the distribution $p(\vec{z})=|c(\vec{z})|^2$. For a good code where code distance scales polynomially with the system size, the codeword state is super-polynomially hard to learn under local SQ. Furthermore, under bit flip noise, such hardness disappears if the noise strength exceeds the decoding threshold, where different logical sectors become indistinguishable. This suggests that learnability can be used as a diagnostic of decoding threshold transition for CSS codes such as good low-density parity check (LDPC) codes.


\emph{Discussions.--}
In this work, we show that nontrivial mixed-state phases of matter, characterized by LI and in some cases long-range CMI, are hard to learn under local SQ and neural network training.  Our predictions are validated in various architectures, such as RNN, CNN and transformer under random initializations.  Thus, hardness of learning serves as a novel probe of mixed-state phases and transitions as well as error-correction thresholds.  

However, there are fine-tuned ways to break the local SQ ansatz and thus make the hard-to-learn states easy to learn. One way is to initialize the neural network in the nontrivial phase; for example, we show in the SM that initializing with the parity distribution makes it easy to learn the SWSSB phase.  The importance of initialization for learnability is also extensively discussed in the ML literature \cite{abbe2025learninghighdegreeparitiescrucial,sutskever13,glorot10a,frankle2019lotterytickethypothesisfinding,he2015delvingdeeprectifierssurpassing}.

Our work focuses on hardness of learning classical distributions. For quantum states, the formalism can be applied to measurement outcomes in a given basis. For example, the measurement outcome distribution of a 1D symmetry protected topological state in the symmetric basis is hard to learn, because the distribution is in the SWSSB phase. One future direction is to extend the formalism  to learn quantum states directly, possibly making use of quantum generalizations of SQ learning.  Our work also explains the hardness of using neural networks to learn distributions at measurement-induced phase transitions \cite{hou2025machinelearningeffectsquantum} and predicts hardness of learning the output distribution of random quantum circuits above a critical depth due to its long-range CMI \cite{PhysRevX.12.021021}.  

Our work suggests fundamental connections between mixed-state phases, error correction, and machine learning that deserve further analysis.  For example, it suggests that the same techniques in error correction for the purpose of hiding data from an environment can also be applied to hide data from machine learning.  Furthermore, it suggests the possibility that an $l$-local Gibbs state may serve as an effective ansatz for the ``simplicity bias'' observed in machine learning.  An interesting avenue to explore is whether and how this simplicity bias manifests in other unsupervised learning schemes, such as diffusion models and normalizing flows. Current research indicates that CMI and phase definitions are pertinent when considering explicitly local flow-based models \cite{hu2025localdiffusionmodelsphases}. However, it is an open question whether simplicity bias renders these concepts relevant in a more general setting. Finally, it would be interesting to apply the CMI diagnostic to real data beyond toy models, e.g., from a noisy quantum simulation, and see how well it quantifies unlearnability.   

{\it Note added:} During the completion of this work, we noticed the recent preprint \cite{cagnetta2026derivingneuralscalinglaws}, which relates the decay of CMI in natural language to neural network scaling laws.

\begin{acknowledgments}
We acknowledge helpful discussions with Samuel Garratt, Tarun Grover, Chong Wang, Yi-Hong Teoh, Schuyler Moss, and Sehmimul Hoque.
This research was supported by Natural Sciences and Engineering Research Council of Canada (NSERC) and an Ontario Early Researcher Award.  Research at Perimeter Institute is supported in part by the Government of Canada through the Department of Innovation, Science and Industry Canada and by the Province of Ontario through the Ministry of Colleges and Universities.
\end{acknowledgments}


\newpage
\onecolumngrid

\appendix

\section{Details of Numerical Experiments}

For the numerical results presented here, we use autoregressive neural network architectures, interpreting the step by step outputs of the model as the conditional log probabilities. To optimize this network, we adopt negative log likelihood as the loss function, taking advantage of the following property of the gradients,
\begin{equation}
    \nabla_\theta D_{KL}(p || p_\theta) = -\nabla_\theta \mathbb{E}_{p} [\log p_\theta].
\end{equation}
The expectation value with respect to $p$ can be approximated through sampling.

\subsection{Network Architectures}

We use the following three architectures in the numerical experiments detailed in this work:

\textbf{LSTM Architecture.} ~\cite{lstm} A three-layer autoregressive recurrent neural network based on Long Short-Term Memory (LSTM) cells. The network processes spin configurations sequentially, with the first two layers using standard LSTM cells that maintain hidden and cell memory states. The final layer consists of a dense layer with ReLU activation that projects to the local Hilbert space dimension, followed by a softmax normalization to produce conditional probabilities for autoregressive sampling. This hybrid design is similar to the architecture in Ref.~\cite{Mohamed_Rnns}.

\textbf{Vanilla Convolutional Architecture.}~\cite{pixelCNN} An autoregressive convolutional neural network using masked convolutions to preserve causality. The network consists of a stack of masked 2D convolutional layers with ReLU activations applied after each layer. Each intermediate layer uses a non-exclusive mask that allows the current site to depend on itself and previous sites, while the final output layer employs an exclusive mask to enforce strict autoregressive ordering where predictions depend only on previously generated sites.

\textbf{Residual Convolutional Architecture.} An autoregressive convolutional neural network employing residual blocks with masked convolutions. The architecture begins with an initial masked convolution to project inputs to the feature dimension, followed by a stack of residual blocks. Each residual block contains two masked convolutional layers with a skip connection, where ReLU activations are applied after each convolution and after the residual addition. Dropout is applied after the first convolution within each residual block. Dropout rate is set to 0.2 for all results in fig.~\ref{fig:toric_syndrome} and for the $L = 4$ data in fig.~\ref{fig:loopnoise} ($L  = 6, 8$ use a dropout of 0.4). The final layer uses an exclusive mask to enforce strict autoregressive generation, ensuring predictions at each site depend only on previously generated sites.

\textbf{Transformer Architecture.}~\cite{Transformer} An autoregressive decoder-only Transformer architecture adapted for many-body physics. The model employs sinusoidal positional encodings to capture spatial information, followed by multiple decoder layers. Each decoder layer contains multi-head self-attention with causal masking to maintain autoregressive properties, followed by a position-wise feedforward network. Residual connections and layer normalization are applied after both the attention and feedforward sublayers. The architecture processes configurations in parallel through the attention mechanism while respecting causal dependencies via a triangular attention mask.

\subsection{Training Data Generation}

With the exception of the 1D Transverse-Field Ising model, we generate a dataset of 100000 samples using Monte-Carlo. For each of those samples, we also calculate the true probability using tensor network methods. For 1D TFIM, we generate samples and calculate probabilities using a matrix-product state (MPS) optimized with density matrix renormalization group (DMRG). For the surface code, we modify the Monte-Carlo samples further. The surface code measurements correspond to $\{1, -1\}$ degrees of freedom that live on the edges of a lattice. We group the two spins on the top and right to each vertex into one degree of freedom of local dimension 4 ($\{-3, -1, 1, 3\}$). For vertices at boundaries that only have one associated spin, we add an artificial $+1$ spin.

\subsection{Optimization Details}

We note that the exact probabilities are not used during the optimization process, but rather only for KL-divergence calculations as a metric of the trained neural network's performance. We use a batch-size of 128 and the Adabelief optimizer with an initial learning rate of $5 \times 10^{-5}$. For each parameter point (system size, noise rate, transverse field), we train up to 30 copies of a neural network with different random initialization. Training runs for 400000 steps. 

\subsection{Evaluation Metrics}

The KL-divergence evaluations use 100000 separate samples, distinct from the training set. We estimate it as $D_{KL}(p || p_\theta) \approx \sum_{x \sim p} \log(\frac{p(x)}{p_\theta(x)})$. The parity metric evaluation involves sampling 10000 configurations from the trained neural network ($x \sim p_\theta$). For the figures in the main text using CNN architectures, we average over the results from different random initializations. However, for the results in App.~\ref{app:B}, we choose the best performing network (the one with the minimal KL). 

\subsection{Implementation Details}

All numerical experiments are implemented using JAX for automatic differentiation and just-in-time compilation, with the NetKet library for quantum many-body physics functionality. The Hilbert space is constructed as a 1D hypercube graph with periodic boundary conditions, even for 2D systems which are serialized into a 1D chain for the autoregressive modeling. For the parameter sweeps, we use noise rates $p \in [0.0, 0.2)$ with a step size of 0.01.

\section{Additional Numerical Results} \label{app:B}
\subsection{Syndrome distribution of repetition code}
\begin{figure}[htbp]
    \includegraphics[width=0.48\textwidth]{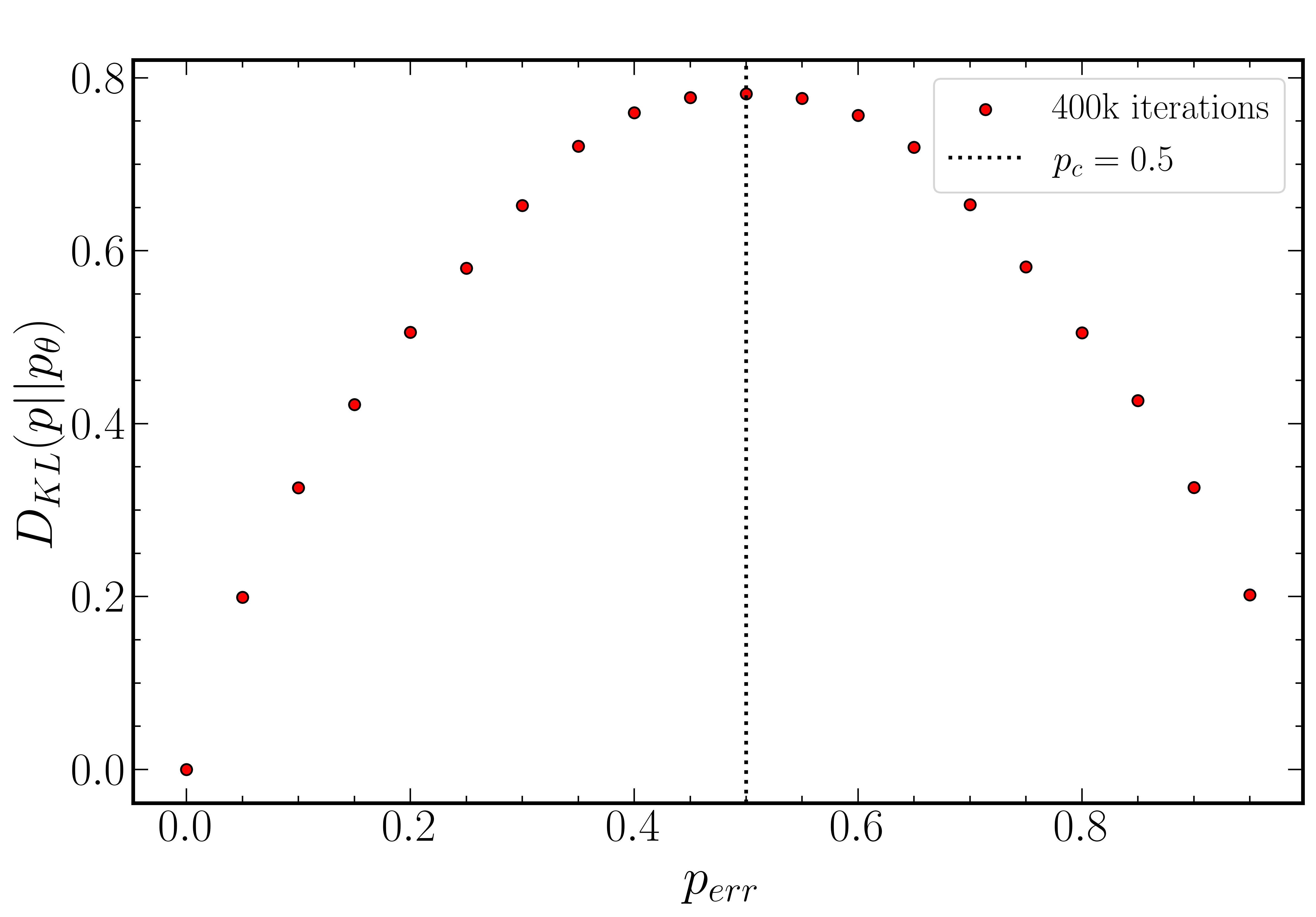}
    \label{fig:GhzKl}
\hspace{2mm}
\caption{\textbf{Performance metrics for 1d syndrome distribution}. KL divergence between trained RNN and exact probabilities of 64 site 1D repetition code syndromes with MPDOs.}
\end{figure}
For the repetition code with error rate $p$ at the threshold $p_c = 1/2$, the syndrome distribution is exactly the parity distribution. This is locally indistinguishable from the maximally mixed state, thus it is hard to learn. On the other hand if $p\neq p_c$, then the Markov length is finite. Thus we expect that only the threshold is hard to learn. Here we have an efficient RNN representation of the distribution due to the existence of the MPDO representation and the MPDO to RNN mapping \cite{Wu_2023_MPSRNN}.

\subsection{Ferromagnetic ground states} 
\begin{figure*}[t]
\begin{subfigure}{0.48 \textwidth}
    \includegraphics[width=\linewidth]{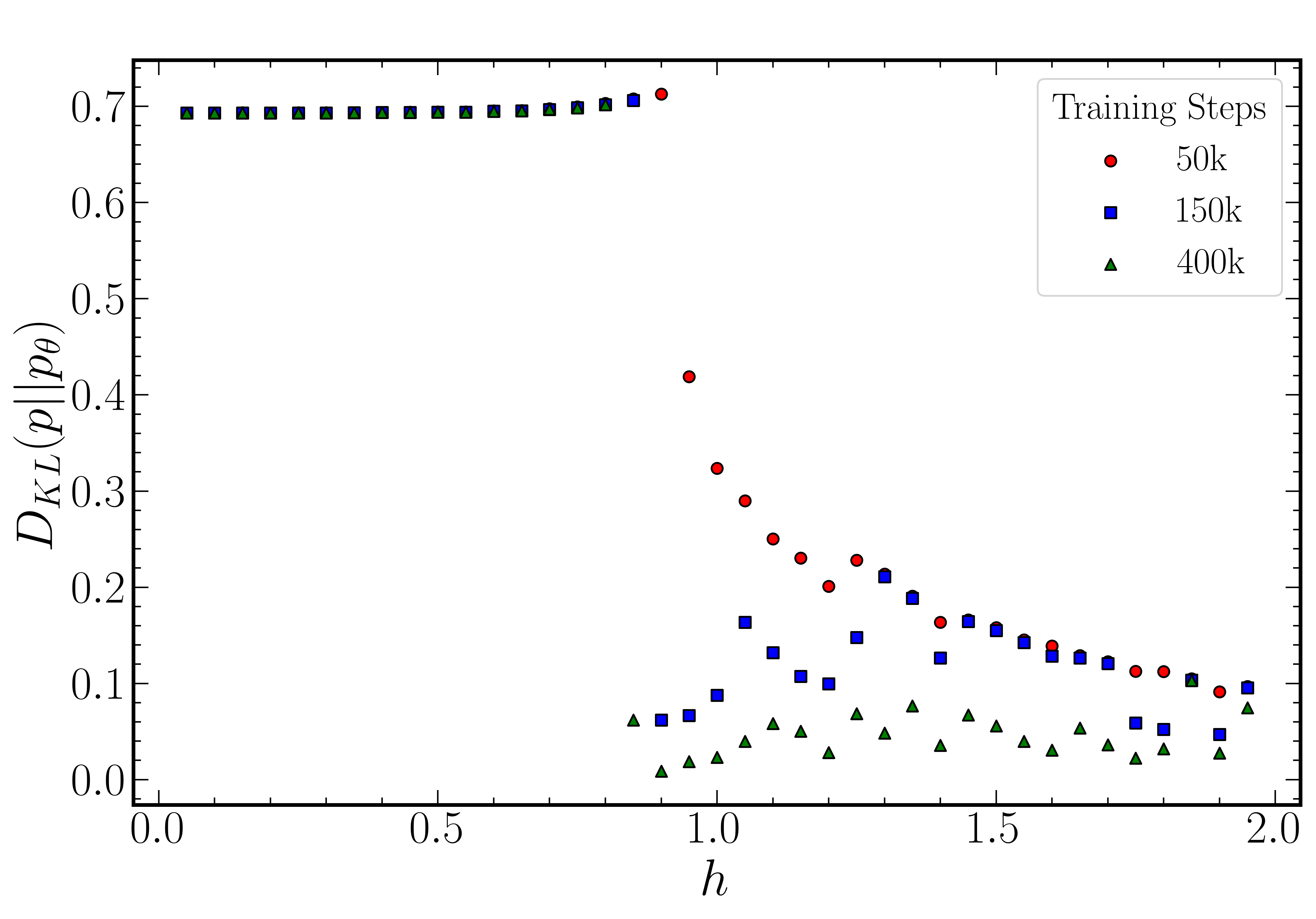}
    \label{fig:IsingKl}
\end{subfigure}
\hspace{2mm}
\begin{subfigure}{0.48 \textwidth}
    \includegraphics[width=\linewidth]{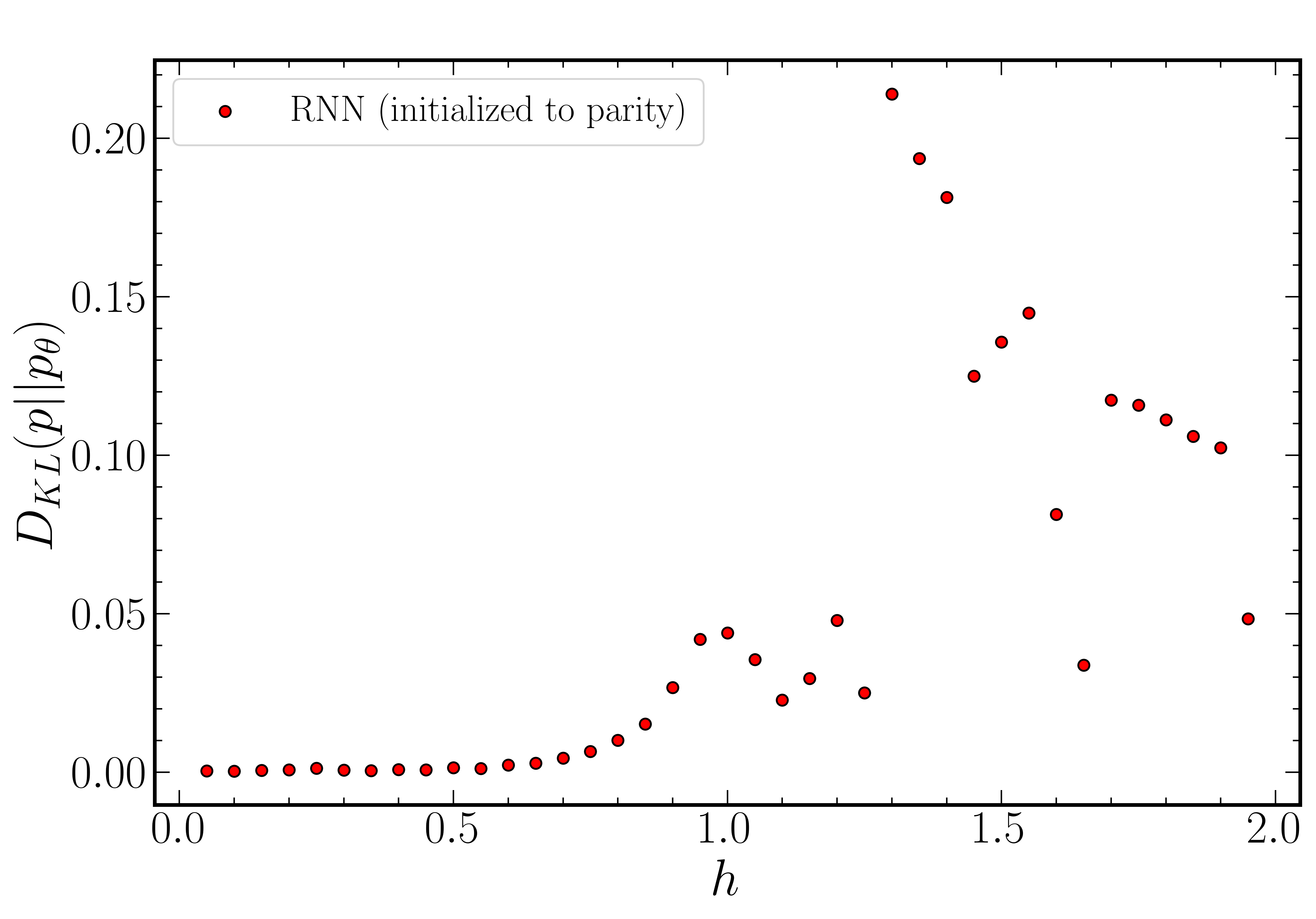}
    \label{fig:IsingMPDOKl}
\end{subfigure}
\caption{\textbf{Performance metrics for the Ising model.} (left) KL divergence between MPS and trained RNN for a 32 site dephased TFIM ground state. KL value remains $\sim \log2$ in the ferromagnetic phase as the RNN misses the one global parity constraint. (right) KL divergence between MPS and trained RNN initialized to parity distribution. As the scale of the y-axis indicates, the in-phase initialization overcomes the hardness of learning.}
\end{figure*}
We also consider the ferromagnetic Ising model:
\begin{equation}  H = - \sum_{\langle i,j \rangle} \sigma_i^x \sigma_j^x - h \sum_i \sigma_i^z.  \end{equation} 
In the ordered phase ($h<1$), the two ground states are SWSSB partners that are locally indistinguishable. Thus, learning is hard. Once again, we expect an efficient RNN representation, because there exists an efficient MPS representation.

\begin{figure*}[htbp]
\begin{subfigure}{0.48 \textwidth}
    \includegraphics[width=\linewidth]{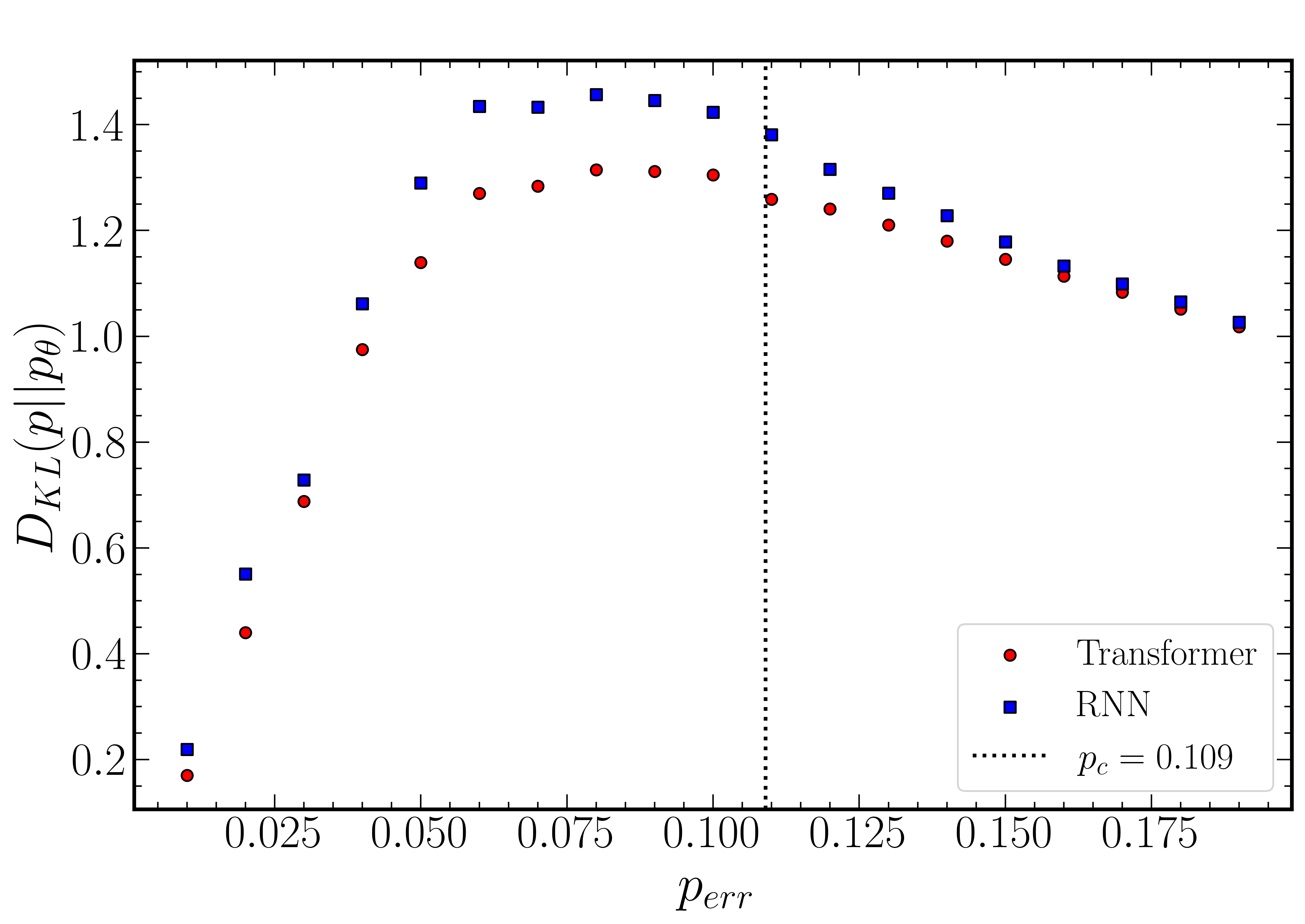}
    \label{fig:SyndromeKlTransformerRNN}
\end{subfigure}
\hspace{2mm}
\begin{subfigure}{0.48 \textwidth}
    \includegraphics[width=\linewidth]{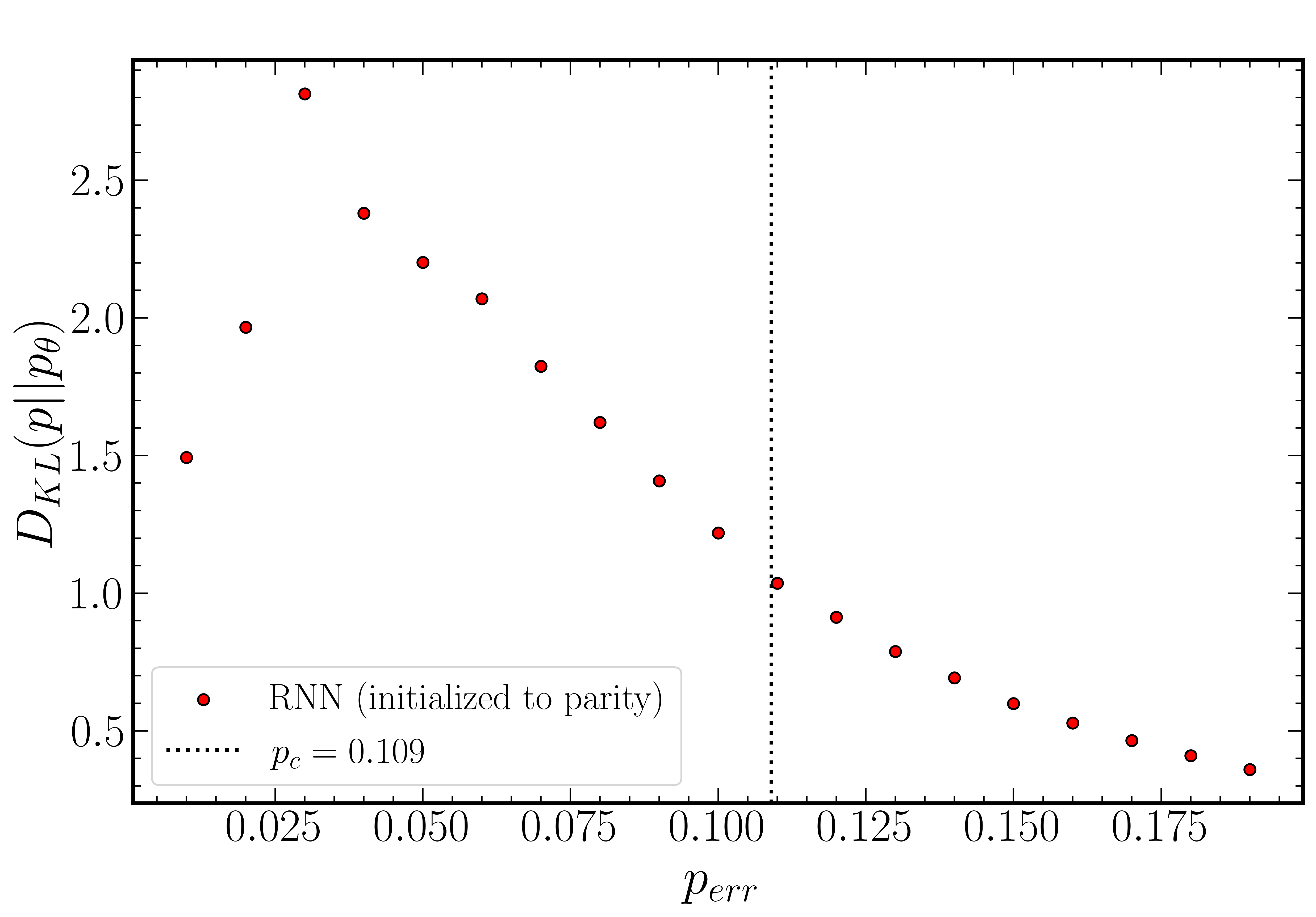}
    \label{fig:SyndromeKLMPDO}
\end{subfigure}
\caption{\textbf{Performance metrics for the 2D Toric code syndromes.} (left) KL divergence between TN and a trained neural network for vertex syndromes of a 7x7 toric code. (right) KL divergence between TN and a trained RNN initialized to parity distribution.}
\end{figure*}
\section{Unlearnability of locally indistinguishable states}
\label{app:C}
Our first main result is that LI states cannot be efficiently learned in the local SQ model. Two states $\rho_1$ and $\rho_2$ are LI if
\begin{equation}
    |\tr((\rho_1 -\rho_2)\phi)| \leq \delta
\end{equation}
for all operators $\phi$ supported on $l_{\phi}$ qubits and $\delta \leq O(e^{-n^{1-c_0}/l_{\phi}})$ for a small constant $c_0>0$. We will provide two proofs of hardness under local SQ learning. The first one is information theoretic, proving that the SQ oracle outputs have exponentially small mutual information with the parameter of the target state. The second one is specific to gradient descent training of neural networks, proving that the training trajectories of $\rho_1$ and $\rho_2$ are exponentially close. While the information-theoretic proof also implies the hardness under the noisy gradient descent setting, the second proof has its own implications on practical neural network training. The result of this section holds for both notions of locality: if locality means $s$-local, the oracles used here are $s$-LSQ; if locality means low weight Fourier-Walsh modes, the oracles used here are $l$-LSQ.

\subsection{Information-theoretic hardness}
To formalize the hardness, we construct a family of states $\rho(p_1) = p_1\rho_1 + (1-p_1) \rho_2$ for $p_1 \in [0,1]$, where $p_1$ is drawn from a distribution $\mathcal{D}$ with constant variance (e.g., $\mathsf{Unif}[0,1]$). Let $\mathcal{A}$ be an $l$-local SQ algorithm making $Q$  queries $\phi_1, \dots, \phi_Q$, each supported on $\leq l$ qubits, receiving answers $\mathbf{v} = (v_1, \dots, v_Q)$. The query function can be adaptive, i.e., $\phi_j$ may depend on the previous query results $v_1,\cdots, v_{j-1}.$ For each query $j$, define $\delta_j = \tr((\rho_1 - \rho_2) \phi_j)$; since $\phi_j$ is $l$-local ($l_{\phi_j} \leq l$), the local indistinguishability implies $|\delta_j|\leq \delta \leq O(e^{-n^{1-c_0}/l})$.

The true expectation for the $j$-th query under $\rho(p_1)$ is
\begin{equation}
    \mu_j(p_1) = \tr(\rho_2 \phi_j) + p_1 \delta_j,
\end{equation}
so the oracle response is $v_j \sim \mathsf{Unif}[\mu_j(p_1) - \tau, \mu_j(p_1) + \tau]$.

The key quantity is the mutual information $I(\mathbf{v}, p_1)$, which upper-bounds the amount of information about $p_1$ (and thus about the underlying state) revealed by the oracle responses $\mathbf{v}$. Adaptivity is handled by
\begin{equation}
    I(\mathbf{v}, p_1) = \sum_{j=1}^Q I(v_j, p_1 | \mathbf{v}_{<j}),
\end{equation}
where $\mathbf{v}_{<j} = (v_1, \dots, v_{j-1})$. 
Conditioned on $\mathbf{v}_{<j}$, the conditional mutual information $I(v_j, p_1|\mathbf{v}_{<j})$ equals $I(p_1,v_j)$, where
\begin{equation}
    v_j = \mu_j(p_1) + Z_j, \quad Z_j \sim \mathsf{Unif}[-\tau, \tau]
\end{equation}
and $Z_j$ is independent of $p_1$. Let $p_1$ be a uniformly random number between $[0,1]$. The mutual information between $p_1$ and $v_j$ is the mutual information between $p_1$ $p_1\delta_j + Z_j$,
\begin{equation}
    I(p_1,v_j| \mathbf{v}_{<j}) = I(p_1,p_1\delta_j + Z_j)  = \frac{\delta_j}{4\tau}
\end{equation}
if $\delta_j < 2\tau$. Summing over the $Q$ queries thus yields
\begin{equation}
    I(\mathbf{v}, p_1) \leq  \frac{Q\delta}{4\tau}.
\end{equation}
For $\delta \leq O(e^{-n^{1-c_0}/l})$ and $\tau^{-1} = \mathrm{poly}(n)$, the mutual information is negligible. Thus, we cannot efficiently learn the value of $p_1$ with local SQ oracles. 
\subsection{Hardness of noisy gradient descent training}
We consider training the variational ansatz $\rho_{\theta}$ using noisy gradient descent with update rule
\begin{equation}
\label{eq:Markov_update}
    \theta_i(t+1) = \theta_i(t) - \gamma\nabla_{\theta_i} S(\rho||\rho_{\theta(t)}) + Z_i(t),
\end{equation}
where $Z_i(t)\sim\mathcal{N}(0,\sigma^2)$ are independent Gaussian noise. The noise is assumed to be polynomially small $\sigma^{-1} = \mathrm{poly}(n)$. Under the local SQ assumption, 
the gradient is the expectation value of local operators $G_i(\theta)$ on $\rho$,
\begin{equation}
    \nabla_{\theta_i} S(\rho||\rho_{\theta(t)}) = \mathrm{tr}\,\rho\, G_i(\theta(t)).
\end{equation}
If the target $\rho$ has a locally indistinguishable partner $\rho_{\mathrm{ref}}$, then
\begin{equation}
\label{eq:LI_LSQDIFF}
    \big|\nabla_{\theta_i} S(\rho||\rho_{\theta(t)}) - \nabla_{\theta_i} S(\rho_{\mathrm{ref}}||\rho_{\theta(t)})\big|
    = \mathrm{tr}((\rho-\rho_{\mathrm{ref}})G_i(\theta(t))) < \delta,
\end{equation}
with $\delta \le O(e^{-n^{1-c_0}/l_G})$. 

The presence of noise renders the updates stochastic, so we regard $\theta_i(t)$ as random variables. To compare the two trainings (one with target state $\rho$ and one with target state $\rho_{\mathrm{ref}}$) we examine the KL divergence between the full trajectory distributions. Let $\mathbb{P}(\{\theta(t), t = 0,1,2\,\cdots T\})$ denote distribution over trajectories $\{\theta(0),\ldots,\theta(T)\}$ with target state $\rho$, and $\mathbb{Q}$ the analogous one driven by target state $\rho_{\mathrm{ref}}$, with identical initialization at $t=0$. We will show that the trace distance between $\mathbb{P}$ and $\mathbb{Q}$ is small under local indistinguishability.

Since the dynamics Eq.~\eqref{eq:Markov_update} are Markovian, we apply the chain rule for KL divergence; writing conditionals at the same current parameter value $\theta$ (i.e. comparing the one-step transition kernels  evaluated at the identical conditioning argument $\theta(t)$) gives
\begin{equation}
\label{eq:trajKL_samecond}
    D_{KL}(\mathbb{P}\|\mathbb{Q})
    = \mathbb{E}_{\{\theta(t)\}\sim\mathbb{P}} \left[\sum_{t=0}^{T-1}
            D_{KL}\big(
                \mathbb{P}(\theta(t+1)\mid\theta(t))
                \,\big\|\,
                \mathbb{Q}(\theta(t+1)\mid\theta(t))
        \big)\right].
\end{equation}

For a fixed conditioning value $\theta$, the two one-step transition kernels are Gaussians with identical covariance $\sigma^2 I$ and means
\[
\mu_{\mathbb{P}}(\theta)=\theta-\gamma\nabla_{\theta}S(\rho\|\rho_{\theta}),\qquad
\mu_{\mathbb{Q}}(\theta)=\theta-\gamma\nabla_{\theta}S(\rho_{\mathrm{ref}}\|\rho_{\theta}).
\]
The KL between these two Gaussians reduces to the squared mean difference:
\begin{equation}
    D_{KL}\big(\mathbb{P}(\theta(t+1)\mid\theta)\,\big\|\,\mathbb{Q}(\theta(t+1)\mid\theta)\big)
    = \frac{\gamma^2\|\nabla_{\theta}S(\rho\|\rho_{\theta})-\nabla_{\theta}S(\rho_{\mathrm{ref}}\|\rho_{\theta})\|^2}{2\sigma^2}.
\end{equation}
Plugging this into Eq.~\eqref{eq:trajKL_samecond} and using the local-indistinguishability bound on gradient differences
\begin{equation}
\label{eq:KLsum_delta}
    D_{KL}(\mathbb{P}\|\mathbb{Q})
    = \sum_{t=0}^{T-1}\mathbb{E}_{\{\theta(t)\}\sim\mathbb{P}}\!\left[
        \frac{\gamma^2\|\nabla_{\theta}S(\rho\|\rho_{\theta(t)})-\nabla_{\theta}S(\rho_{\mathrm{ref}}\|\rho_{\theta(t)})\|^2}{2\sigma^2}
    \right]
    \le \frac{T\gamma^2\delta^2}{2\sigma^2}.
\end{equation}

The KL divergence of any marginal law at time $t$ is upper bounded by the KL of the full trajectories due to data processing inequality:
\begin{equation}
    D_{KL}\big(\mathbb{P}_{\theta(t)}\,\|\,\mathbb{Q}_{\theta(t)}\big)
    \le D_{KL}(\mathbb{P}\|\mathbb{Q})
    \le \frac{T\gamma^2\delta^2}{2\sigma^2}.
\end{equation}
Applying Pinsker's inequality yields the total variation bound on the marginals:
\begin{equation}
    \frac{1}{2}||\mathbb{P}_{\theta(t)} -\mathbb{Q}_{\theta(t)}||_1
    \le \sqrt{\tfrac{1}{2}D_{KL}(\mathbb{P}\|\mathbb{Q})}
    \le \frac{\gamma\delta\sqrt{T}}{2\sigma}.
\end{equation}
Since $\delta\le O(e^{-n^{1-c_0}/l_G})$, the total variation distance between the parameter distributions for any polynomially bounded $T,\gamma$ and $\sigma^{-1}$ remains exponentially small.

Therefore, we have shown that samples from two locally indistinguishable distributions give rise to exponentially close training trajectories under noisy gradient state training with local SQ oracles. 

\section{Stability of local indistinguishability}
Suppose $\sigma_1$ and $\sigma_2$ are LI, and  $||\sigma_1 - \mathcal{D}(\rho_1)||\leq \epsilon$ and $||\sigma_2 - \mathcal{D}(\rho_2)||\leq \epsilon$. We can establish that $\rho_1$ and $\rho_2$ are globally distinguishable. This comes from $||\sigma_1 -\sigma_2||_1 \geq C'$ and the one way channel connection, . We have
\begin{equation}
    \begin{aligned}
        C' &\leq||\sigma_1 -\sigma_2 ||_1 \\&\leq ||\sigma_1 - \mathcal{D}(\rho_1)||_1 + ||\mathcal{D}(\rho_1)- \mathcal{D}(\rho_2)||_1 + ||\sigma_2 - \mathcal{D}(\rho_2)|| \\
        & \leq 2\epsilon + ||\mathcal{D}(\rho_1)- \mathcal{D}(\rho_2)||_1 \\
        & \leq 2\epsilon + ||\rho_1 -\rho_2||_1
    \end{aligned}
\end{equation}
by triangle inequality and data processing inequality, thus $||\rho_1-\rho_2||\geq C' -2\epsilon$. 

The reverse direction $||\rho_1 - \mathcal{N}(\sigma_1)||\leq \epsilon$ and $||\rho_2 - \mathcal{N}(\sigma_2)||\leq \epsilon$, gives LI of $\rho_1$ and $\rho_2$. Let the channel $\mathcal{N}$ be depth $d$, i.e., it can be dilated into a unitary circuit $U_{SA}$ of depth $d$ such that $\mathcal{N}(\rho_S) = \tr_A( U_{SA}(\rho_S \otimes |0\rangle\langle 0|_A)U^{\dagger}_{SA}) $. Then for any $l_{\phi}$-local operator $\phi_S$,
\begin{equation}
\begin{aligned}
    \tr_S[(\rho_1-\rho_2)_S\phi_S] &= \tr_{SA}[U((\sigma_1-\sigma_2)_S\otimes |0\rangle \langle 0|_A)U^{\dagger})\phi_S] \\
    &= \tr_S [(\sigma_1 -\sigma_2)_S \mathcal{N}^{*}(\phi_S)]
\end{aligned}
\end{equation}
where 
\begin{equation}
    \mathcal{N}^{*}(\phi_S) = \langle 0_A| U^\dagger (\phi_S \otimes I_A) U |0_A\rangle
\end{equation}
The support of $\mathcal{N}^{*}(\phi_S)$ is at most $2d\times l_{\phi}$, when the support of $\phi_S$ is spatially disjoint. Also note that the spectral norm $||\mathcal{N}^{*}(\phi_S)|| \leq || U^\dagger (\phi_S \otimes I_A) U || = ||\phi_S|| $ does not blow up. Thus,
\begin{equation}
    |\tr_S[(\rho_1-\rho_2)_S\phi_S]|  \leq O(e^{-n^{1-c_0}/(2dl_{\phi})})
\end{equation}
If $d = \text{polylog} (n)$, we still have
\begin{equation}
    |\tr_S[(\rho_1-\rho_2)_S\phi_S]|  \leq O(e^{-n^{1-c'_0}/l_{\phi}})
\end{equation}
for any $c'_0>c_0$. Therefore, $\rho_1$ and $\rho_2$ are locally indistinguishable.
\section{Local SQ learning of finite Markov length states}
\label{sec:AppE}
In this section prove the easiness of learning classical distributions with finite Markov length. Since we will deal with multiple partitions, we distinguish two notions of locality: $l$-local SQ ($l$-LSQ) means that the query function has at most weight-$l$ terms in the Fourier-Walsh expansion; $l$-spatially-local SQ ($l$-SLSQ) means that the query function has support on at most $l$ sites on a geometrically local patch. 
\subsection{Distributions on 1D lattice}
In this subsection we prove that a probability distribution on a 1D open chain with finite Markov length $\xi$ can be efficiently learned with SLSQ oracles. Then we generalize the result to periodic boundary conditions. For the open chain, we have the theorem formalized below.
\begin{figure}
    \centering
    \includegraphics[width=0.9\linewidth]{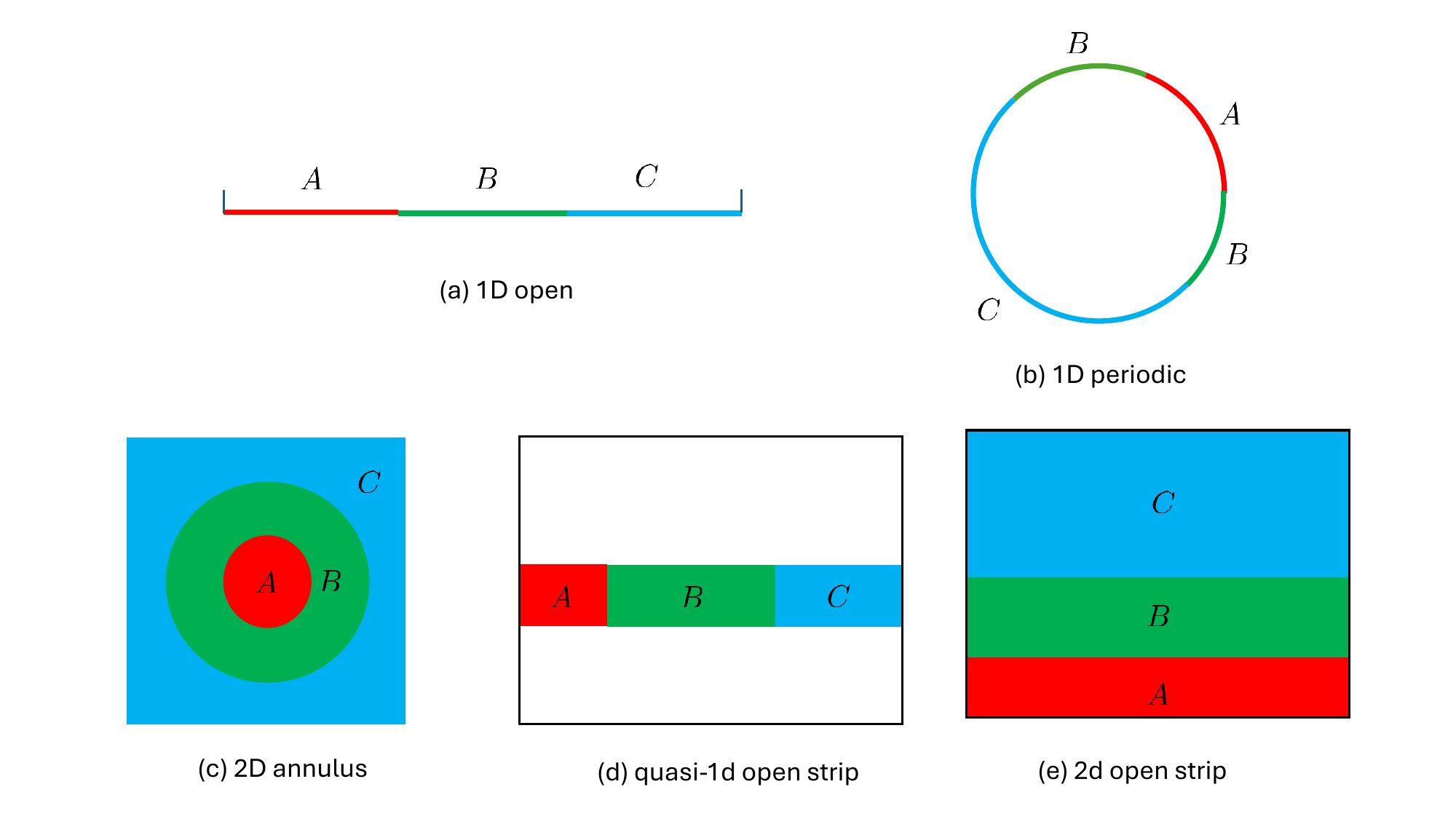}
    \caption{Partitions of CMI $I(A:C|B)$ for 1D and 2D systems. For (b), we assume for simplicity that the two intervals for $B$ are of equal length. For (c), we assume for simplicity that $A$ and $AB$ are concentric balls.}
    \label{fig:partitions}
\end{figure}
\begin{theorem}
\label{thm:appSQCMI}
Let $p$ be a distribution of $n$ spins on 1D with finite Markov length $\xi$, i.e., $I(A:C|B)\leq \mathrm{poly}(|A|,|C|) e^{-|B|/\xi}$ for any contiguous intervals $A,B,C$ (see the partition (a) in Fig.~\ref{fig:partitions}). Let $x>2$ be a positive constant, then there exists a $O(\log n)$-spatially-local SQ algorithm using $M_{tot} \leq O(n^{2x\xi + 1})$ SLSQ oracles with tolerance $\tau = O(n^{-x(2\xi + 1/2) +c_0})$ such that it produces a probability distribution $\tilde{p}$ with error $||p-\tilde{p}||_1 \leq O(n^{-x/2 + 1 + c_0})$, with $c_0$ being an arbitrarily small positive constant.
\end{theorem}
\begin{proof}
We begin with a lemma which reconstructs local marginals from local SQ oracles.
\begin{lemma}
A probability distribution $p$ of $L$ bits can be approximately obtained to error $\epsilon$ with $M = 2^{L}$ $L$-local SQ oracles with tolerance $\tau =  2^{-L} \epsilon$. If the L bits are spatially contiguous, then the oracles are SLSQ.
\end{lemma}
\begin{proof}
Consider the query functions $\phi_{S}(s) = \prod_{j\in S} s_j $ where $S$ is a subset of $[L]:=\{1,2,\cdots L\}$. This forms a complete orthonormal set of Boolean functions and thus
\begin{equation}
    p(s) = \frac{1}{2^L}\sum_{S} \langle \phi_S \rangle_p \phi_S(s),
\end{equation}
where $\langle \phi_S \rangle_p = \mathbb{E}_{s\sim p} \phi_S(s)$ is the Fourier-Walsh coefficient. The SQ oracle returns $v_{S}$ which
\begin{equation}
    |v_{S} - \langle \phi_S \rangle_p| < \tau
\end{equation}
We can thus reconstruct an approximated probability distribution
\begin{equation}
    \tilde{p}(s) = \frac{1}{2^L}\sum_{S} v_S \phi_S(s)
\end{equation}
The estimation error is upper bounded by
\begin{equation}
\begin{aligned}
    ||p- \tilde{p}||_1 &=  \frac{1}{2^L} \sum_s \left|\sum_S( v_S - \langle \phi_S\rangle_p) \phi_S(s) \right| \\
    &\leq \frac{1}{2^L} \sum_s \sum_S |v_S - \langle \phi_S\rangle_p| |\phi_S(s)| \\
    &\leq \frac{1}{2^L} \sum_s \sum_S \tau \cdot 1 \\
    &= 2^L \tau
\end{aligned}
\end{equation}
Therefore we use these $2^L$ query functions with tolerance $\tau =  2^{-L} \epsilon$ to each an error $||p-\tilde{p}||_1 \leq \epsilon$. 
\end{proof}

For a 1d distribution $p$ on $n$ spins, we divide it into $m = n/l$ intervals $A_1,\cdots A_m$ with length $l= x\xi \log_2 n$, where $x$ is a positive constant and $\xi$ is the Markov length. For each of the subsystems $A_i A_{i+1}$, we run the $L$-local SQ algorithm with $L=2l$ above. The total number of oracles is 
\begin{equation}
    M_{tot} = m 2^{2l} \leq n^{2x\xi + 1}
\end{equation}
and the tolerance is 
\begin{equation}
\label{eq:tau}
    \tau = 2^{-2l} \epsilon = n^{-2x\xi} \epsilon
\end{equation}
This achieves a small error $\epsilon$ on the marginals,
\begin{equation}
\label{eq:app_markov_marginal}
    ||p_{A_i A_{i+1}} - \tilde{p}_{A_i A_{i+1}} ||_1 \leq \epsilon.
\end{equation}
Finally we construct an approximation to the true distribution with the learned marginals $\tilde{p}_{A_i A_{i+1}}$ through the Markov distribution
\begin{equation}
    \tilde{p}_{A_1\cdots A_m} =  \tilde{p}_{A_1 A_2} \prod_{i=2}^{m-1} \tilde{p}_{A_{i+1}|A_i}.
\end{equation}
We will bound the error $||p-\tilde{p}||_1$ from the triangle inequality. We define another Markov distribution
\begin{equation}
 q_{A_1\cdots A_m} =  p_{A_1 A_2} \prod_{i=2}^{m-1} p_{A_{i+1}|A_i}
\end{equation}
by combining marginals from the true distribution. Following theorem B.1 of Ref.~\cite{yang_cmi_nns}, $q$ is an approximation to $p$ with small error,
\begin{equation}
    ||p-q||_1 \leq O(n^{-x/2 +c_0 + 1})
\end{equation}
for arbitrarily small positive constant $c_0>0$.

It remains to bound the distance between the two exact Markov distributions $\tilde{p}$ and $q$. We need the following lemma:
\begin{lemma}
\label{lemma:markov_chain_inequality}
    Given two Markov probability distributions $p_{ABC} = p_{AB}p_{C|B}$ and $q_{ABC}= q_{AB}q_{C|B}$, if $||p_{AB}-q_{AB}||_1 \leq \epsilon$ and $||p_{BC}-q_{BC}||_1 \leq \epsilon$, then $||p_{ABC}-q_{ABC}||_1 \leq 3\epsilon $. 
\end{lemma}
\begin{proof}
    Define $r_{ABC} = p_{AB}q_{C|B}$, then $||q_{ABC}-r_{ABC}||_1 = || (p_{AB}-q_{AB}) q_{C|B} ||_1 = ||p_{AB}-q_{AB}||_1 \leq \epsilon$ and
    \begin{equation}
        \begin{aligned}
            ||p_{ABC}-r_{ABC}||_1 &= ||p_{AB}(p_{C|B}-q_{C|B})||_1  \\
            &= ||p_{B}(p_{C|B}-q_{C|B})||_1  \\
            &= ||p_{BC}-p_Bq_{C|B}||_1  \\
            &\leq ||p_{BC}-q_{BC}||_1 + ||q_{BC} - p_Bq_{C|B}||_1  \\
            &= ||p_{BC}-q_{BC}||_1 + ||q_{B} - p_B||_1 \\
            &\leq 2||p_{BC}-q_{BC}||_1 = 2\epsilon
        \end{aligned}
    \end{equation}
    where we have used the triangle inequality in the 4th line and data processing inequality in the last line. Finally, $||p_{ABC}-q_{ABC}||_1\leq ||p_{ABC}-r_{ABC}||_1 + ||q_{ABC}-r_{ABC}||_1 \leq 3\epsilon$ by triangle inequality.
\end{proof}
Generalization of the lemma to $m$-parties, we have
\begin{equation}
\label{eq:markov_chain_inequality}
    ||\tilde{p}-q||_1 \leq (2m-3)\epsilon.
\end{equation}
To prove it, we introduce a series of auxiliary distributions 
\begin{equation}
    h_k = q_{A_1\cdots A_{k+1}} \prod_{j=k+1}^{m-1} \tilde{p}_{A_{j+1}|A_j}
\end{equation}
with $h_0 = \tilde{p}$ and $h_{m-1} = q$. We thus have
\begin{equation}
    ||\tilde{p}- q||_1 \leq \sum_{k=0}^{m-2} ||h_{k}-h_{k+1}||_1
\end{equation}
Following a similar proof as Lemma~\ref{lemma:markov_chain_inequality}, we have $||h_0-h_1||_1 \leq \epsilon$ and $||h_k-h_{k+1}||_1 \leq 2\epsilon$ ($k\geq 1$). Summing up the errors give Eq.~\eqref{eq:markov_chain_inequality}. 

Finally, the triangle inequality then implies that
\begin{equation}
    ||p -\tilde{p}||_1 \leq ||p-q||_1 + ||\tilde{p}-q||_1 
\end{equation}
We choose $\epsilon = O(n^{-x/2 + c_0})$, such that the second term is on the same order as the first term. Thus,
\begin{equation}
    ||p -\tilde{p}||_1 \leq O(n^{-x/2 +c_0 + 1})
\end{equation}
Using Eq.~\eqref{eq:tau} we need the tolerance to be $\tau = O(n^{-x(2\xi + 1/2) +c_0})$. 
\end{proof}
We can generalize the theorem to 1D periodic boundary conditions straightforwardly.
\begin{corollary}
    For a distribution $p$ on a 1D periodic chain with $n$ spins,  \\
    (1) If both for 1D periodic partitions in Fig.~\ref{fig:partitions}(b) and the 1D open partition for contiguous intervals $ |ABC|\leq n-O(\log n)$ that do not cover the whole system, the CMI decays exponentially with $|B|$, then there exists an efficient SLSQ algorithm that learns $p$. \\
    (2) If for the 1D periodic partition in Fig.~\ref{fig:partitions}(b) the CMI decays exponentially with $|B|$, then there exists an efficient LSQ algorithm that learns $p$.
\end{corollary}
\begin{proof}
    To prove (1), we first use theorem~\ref{thm:appSQCMI} to learn $p_{A_1 A_2\cdots A_{m-1}}$ and $p_{A_{m-1} A_m A_1}$ with $\mathrm{poly}(n)$ GLSQ oracles. The algorithm outputs $\tilde{p}_{A_1 A_2\cdots A_{m-1}}$ and $\tilde{p}_{A_{m-1} A_m A_1}$ that approximate the two distributions, with
    \begin{equation}
        ||p_{A_1 A_2\cdots A_{m-1}}-\tilde{p}_{A_1 A_2\cdots A_{m-1}}||_1 = \mathrm{poly}(n)^{-1},~~ ||p_{A_{m-1} A_m A_1}-\tilde{p}_{A_{m-1} A_m A_1}||_1 = \mathrm{poly}(n)^{-1}.
    \end{equation}
We then construct
\begin{equation}
    \tilde{p}_{A_1 A_2 \cdots A_m} = \tilde{p}_{A_1 A_2\cdots A_{m-1}} \tilde{p}_{A_m|A_{m-1}, A_1}
\end{equation}
which is an approximation to $q = p_{A_1 A_2\cdots A_{m-1}} p_{A_m|A_{m-1}, A_1}$. The error $||\tilde{p}-q||_1$ is also polynomially small in $n^{-1}$ due to Lemma~\ref{lemma:markov_chain_inequality}. Furthermore
\begin{equation}
    ||p-q||_1 \leq \sqrt{2D_{KL}(p||q)} = \sqrt{2 I(A_2\cdots A_{m-2}:A_m|A_1,A_{m-1})} =\mathrm{poly}(n)^{-1}
\end{equation}
given the exponential decay of CMI in the 1D periodic partition with $|A_1| = |A_{m-1}| = O(\log n)$. By the triangle inequality $||p-\tilde{p}||_1$ is also polynomially small.

To prove (2), we perform state tomography to obtain $\tilde{p}_{A_{m} A_1 A_2}$ and $\tilde{p}_{A_i A_{i+1}A_{m-i+1}A_{m-i+2}}$ for $i=2, 3 \cdots m/2$. The first only requires SLSQ but the latter requires LSQ because it is a disjoint union of two intervals. We then reconstruct
\begin{equation}
    \tilde{p}_{A_1 \cdots A_m} = \tilde{p}_{A_{m} A_1 A_2} \prod_{i=1}^{m/2} \tilde{p}_{ A_{i+1}A_{m-i+1}|A_i,A_{m-i+2}}
\end{equation}
and the same proof of theorem~\ref{thm:appSQCMI} bounds the distance $||p-q||_1$ to be polynomially small.
\end{proof}

\subsection{Distributions on higher-dimensional lattice}
We briefly comment on generalization of Theorem~\ref{thm:appSQCMI} to higher dimensions.

For a 2D lattice, the partition can be much more diverse. The most commonly used Markov length is defined for the annulus partition (Fig.~\ref{fig:partitions}(c)). But finite Markov length in this partition does not imply that the state is efficiently learnable under LSQ. In fact, the counterexample is the loop model $\rho_{l_1 l_2}$ with logical constraints $l_1 = \pm1$ or $l_2=\pm1$ (see appendix~\ref{sec:genSQ} subsection 3 and Fig.~\ref{fig:LoopModelDiagram}). The state has zero Markov length but is hard to learn under LSQ due to LI with the loop model with no logical constraints.

Here we adopt the strip partition in Fig.~\ref{fig:partitions}(d) and (e). For all the quasi-1d strips of width $O(\log n)$, we require the CMI to be exponentially decaying in the length of $B$. Then one can learn the marginal distributions on the strips by multiplying marginals on $O(\xi\log n) \times O(\xi\log n)$ squares with polynomially small errors. Our proof of theorem~\ref{thm:appSQCMI} then directly generalizes. However, the full tomography of the squares cost $2^{O((\xi\log n)^2)}$, which is quasi-polynomial instead of polynomial. As a result, the marginal distributions on the strips can be learned with a quasi-polynomial number of SLSQ oracles to get a polynomially small error in trace distance. Finally, we require the CMI exponentially decaying in the width of $B$ for the the partition in Fig.~\ref{fig:partitions}(e). This allows us to glue the marginals on the strips to obtain an approximation to the full distribution $p$. To summarize, under the finite Markov length condition for both partitions (d) and (e), we have an algorithm to learn the full distribution with quasi-polynomial number of SLSQ oracles and polynomially small error.

Before moving on, we comment on the partition in Ref.~\cite{yang_cmi_nns}. We note that Ref.~\cite{yang_cmi_nns} adopts a partition determined by a 1d snake-like path (the Fig. 2(b) of Ref.~\cite{yang_cmi_nns}), which is distinct from both the annulus partition and the strip partition. If the CMI decays exponentially for this partition, then the distribution $p$ can again be approximated by multiplying marginals on $O(\xi\log n) \times O(\xi\log n)$ squares with polynomially small errors. Therefore, it also takes a quasi-polynomial number of SLSQ oracles.

\section{Long-range CMI and local indistinguishability}
In this section, we show the relation between long-range CMI of a classical distribution and existence of locally indistinguishable partner. We will first prove the result on tripartite systems and then apply it to the many-body setting.
\subsection{Tripartite systems}
\begin{theorem}
\label{thm:LICMI}
    Let $p_{ABC}$ be a tripartite classical distribution. The following two statements hold:\\
    (1) If $I(A:C|B) = \epsilon > 0$, then there exists a locally indistinguishable state $q_{ABC}$ such that $p_{AB} = q_{AB},~ p_{BC}=q_{BC}$ and $||p-q||_1 \geq \min\{1, T_0(\epsilon)\}$, where $T_0$ is the solution to $4T \log d_{min} + 2H(T) = \epsilon$ with $H(T):= -T\log T - (1-T) \log (1-T)$ and $d_{min}:=\min\{d_A, d_C\}$. \\
    (2) If there exists a locally indistinguishable state $q_{ABC}$ such that $||p_{AB} - q_{AB}||_1<T',~ ||p_{BC}-q_{BC}||<T'$ and $||p-q||_1 = T>3T'$, then at least one state of $p_{ABC}$ and $q_{ABC}$ has a CMI lower bound $I(A:C|B)\geq \frac{1}{8}(T-3T')^2$.
\end{theorem}
\begin{proof}
    To prove (1), we consider the Markov distribution $q_{ABC} = p_{AB} p_{C|B}$. The marginals are the same as $p_{AB} = q_{AB}$ follows from $\sum_c p_{C|B}(c|b) = 1$ and $p_{BC} = q_{BC}$ follows from $\sum_{a} p_{AB}(a,b) = p_B(b)$. To prove the trace distance, we note that $I(A:C|B)_p = I(A:C|B)_q -S({ABC})_p+S({ABC})_q$. Since $I(A:C|B)_q = 0$, we have $S(ABC)_q-S(ABC)_p = \epsilon$. Note that $S(AB)_q = S(AB)_p$, we therefore have $S(C|AB)_q - S(C|AB)_p = \epsilon$. Now we use the Alicki-Fannes inequality for the case of $T = ||p-q||_1\leq 1$: 
    \begin{equation}
       |S(C|AB)_p -S(C|AB)_q|\leq 4 T \log d_C + 2H(T) 
    \end{equation}
    If $T\leq 1$ , this gives a lower bound of $T \geq T_0 (\epsilon)$, where $T_0$ is the smallest solution to $4 T \log d_C + 2H(T) = \epsilon$. We may alternatively use $|S(A|BC)_p - S(A|BC)_q| =\epsilon$, which gives another bound with $d_C$ and $d_A$ exchanged. Therefore, if either $d_A$ and $d_C$ are $O(1)$ and $T\leq 1$, then $T \geq T_0 (\epsilon)$ where $T_0$ is the solution to $4T \log d_{min} + 2H(T) = \epsilon$. If $T>1$, the bound on $||p-q||_1$ is trivially satisfied.

    To prove (2), suppose that both $p$ and $q$ satisfy $I(A:C|B) <  \epsilon$. Define $\tilde{p}_{ABC} = p_{AB} p_{C|B}$ and $\tilde{q}_{ABC}= q_{AB} q_{C|B}$. Then $D_{KL}(p||\tilde{p}) = I(A:C|B)_{p} < \epsilon$ and $D_{KL}(q||\tilde{q}) = I(A:C|B)_{q} < \epsilon$. Using the Pinsker inequality
    \begin{equation}
        D_{KL}(p||q) \geq \frac{1}{2} ||p-q||^2_1
    \end{equation}
    we obtain $||p-\tilde{p}||_1 < \sqrt{2\epsilon}$ and $||q-\tilde{q}||_1 < \sqrt{2\epsilon}$. Using Lemma~\ref{lemma:markov_chain_inequality}, we have $||\tilde{p}-\tilde{q}||_1 \leq 3T$.
    The triangle inequality implies that $ T =||p-q||_1 \leq ||p-\tilde{p}||_1+ ||q-\tilde{q}||_1 + ||\tilde{p}-\tilde{q}||_1  < 2\sqrt{2\epsilon} + 3T' $. If $\epsilon \leq \frac{1}{8} (T-3T')^2$, we would get a contradiction. Thus, one of the CMIs in $p_{ABC}$ and $q_{ABC}$ must be at least $\frac{1}{8}(T-3T')^2$.
\end{proof}
\subsection{Distributions on 1D open chain}
We can apply the theorem \ref{thm:LICMI} straightforwardly to a 1d distribution of $n$ spins on an open chain. We choose $A$ to be spins 1 to $l_A$, $B$ to be spins $l_A+1$ to $l_A+l_B$ and $C$ is the rest of the system (see Fig.~\ref{fig:partitions}(a)). We have the following corollary of theorem \ref{thm:LICMI}.
\begin{corollary}[Local indistinguishability of long-range CMI distributions]
    If the classical distribution $p_{ABC}$ on $n$ spins has $I(A:C|B)  = O(1)$  for $l_A = O(1)$ and $l_B = O(n)$, then there exists another distribution $q_{ABC}$ such that it is $s$-locally indistinguishable from $p_{ABC}$ on any contiguous $l_B$ spins but globally distinguishable with $||p-q||_1 = O(1)$. 
\end{corollary}
\begin{proof}
    By applying theorem \ref{thm:LICMI}, we have $p_{AB} = q_{AB}$ and $p_{BC} = q_{BC}$. Any contiguous $l_B$ spins either belongs to $AB$ or $BC$, thus $p_{ABC}$ and $q_{ABC}$ cannot be distinguished from the marginals on these spins. Finally, both $\log d_A = l_A = O(1)$ and $\epsilon = O(1)$ in $4T\log d_A + 2H(T) = \epsilon$, implying that the solution $T_0 = O(1)$. Therefore, $||p-q||_1 = O(1)$ and the states are globally distinguishable. 
\end{proof}
The case of periodic boundary condition can be directly generalized, where we require $B$ to be the union of $l_B$ spins to the right of $A$ and $l_B$ spins to the left of $A$.

For the converse statement, if we have two locally indistinguishable states on $O(n)$ spins, then one of them must have long-range CMI. 
\begin{corollary}[Long-range CMI for locally indistinguishable states]
    Given two classical distributions $p$ and $q$ of $n$ spins on an open chain. If they are $s$-locally indistinguishable such that $||p_I-q_I||_1 \leq T_l$ on any contiguous interval $I$ of $l$ spins, and $||p-q||_1  = T > \frac{4T_l n}{l}$, then for either $p$ or $q$ there exists contiguous intervals $A,B,C$ such that $I(A:C|B)\geq \frac{1}{2}(\frac{Tl}{2n}-2T_l)^2$ and $|B| = l$. If $l=O(n)$, this puts an $O(1)$ lower bound on the long-range CMI. 
    \label{thm:LI_imply_CMI}
\end{corollary}
\begin{proof}
    We again prove by contradiction. Suppose that for all contiguous $A,B,C$ with $|B|=l$ we have $I(A:C|B)<\epsilon$ for both distributions.
    Let $A_k = [1 + ka,(k+1)a]$ be a interval of $l$ spins and $a=l/2$. Then we define $\tilde{p} = p_{A_1 A_2} p_{A_3|A_2} p_{A_4|A_3} \cdots p_{A_m|A_{m-1}}$ and $\tilde{q} = q_{A_1 A_2} q_{A_3|A_2} q_{A_4|A_3} \cdots q_{A_m|A_{m-1}}$, where $m=n/a = 2n/l$. The $m$-party generalization of lemma~\ref{lemma:markov_chain_inequality} implies that
    \begin{equation}
        ||\tilde{p}-\tilde{q}||_1 < (2m-3) T_l <2mT_l
    \end{equation}
    To bound the distance between $p$ and $\tilde{p}$, we make sequential use of lemma~\ref{lemma:markov_chain_inequality} and the conditions of small CMI. The marginals on the first three parties are close because 
    \begin{equation}
        ||p_{A_1 A_2 A_3} - \tilde{p}_{A_1 A_2 A_3}||_1 \leq \sqrt{2 D_{KL}(p_{A_1A_2A_3}||\tilde{p}_{A_1A_2A_3})} \leq \sqrt{2\epsilon}
    \end{equation}
    Going to the first four parties, we use
    \begin{equation}
    \begin{aligned}
         ||p_{A_1 A_2 A_3 A_4} - \tilde{p}_{A_1 A_2 A_3 A_4}||_1 &\leq ||p_{A_1 A_2 A_3 A_4} - p_{A_1 A_2 A_3} p_{A_4|A_3}||_1 +  ||p_{A_1 A_2 A_3} p_{A_4|A_3} - \tilde{p}_{A_1 A_2 A_3} p_{A_4|A_3}||_1 \\
         &\leq \sqrt{2\epsilon} + ||p_{A_1 A_2 A_3} - \tilde{p}_{A_1 A_2 A_3}||_1 \\
         &\leq 2\sqrt{2\epsilon} 
    \end{aligned}
    \end{equation}
    Going to $m$-parties, we have
    \begin{equation}
        ||p-\tilde{p}||_1 \leq (m-2)\sqrt{2\epsilon} < m\sqrt{2\epsilon}
    \end{equation}
    Therefore, $ T = ||p-q||_1 \leq  ||p-\tilde{p}||+||q-\tilde{q}|| + ||\tilde{p}-\tilde{q}||_1 < m(\sqrt{2\epsilon} + 2T_l)$. One would get a contradiction if $\epsilon<\frac{1}{2}(\frac{T}{m}-2T_l)^2$, which completes the proof.
\end{proof}

\subsection{Distributions on higher-dimensional lattice}
Finally, we comment on the generalizations to higher dimensions. 

The direction of long-range CMI implying $s$-LI can be directly generalizes to annulus partition (Fig.~\ref{fig:partitions}(c)) in any dimensions. Let $A$ and $AB$ be concentric balls with the radius of $A$ being $O(1)$, and the radius difference between $A$ and $AB$ to be $l_B$. If $I(A:C|B) = O(1)$, then we can construct a state $q_{ABC} = p_{AB} p_{C|B}$ that is $s$-LI with $p_{ABC}$ on any ball of diameter $l_B$. 

As we noted in the subsection 2 of appendix~\ref{sec:AppE}, the reverse direction (LI implying long-range CMI) does not hold on the annulus partition in 2D. However, for the strip partition (Fig.~\ref{fig:partitions}(e)), we can generalize the direction of $s$-LI implying long-range CMI. This can be easily seen because that Corollary \ref{thm:LI_imply_CMI} does not depend on the local Hilbert space dimension.


\section{Justification of local SQ assumption}
\label{sec:AppG}
In this section we derive in more detail the condition that unsupervised learning of a distribution is described by a local SQ model. In the first subsection, we show that LSQ model describes the noisy gradient descent training of LI states, given that one of the states is quasi-local Gibbsian. The result can be directly applied to SLSQ model if notion of locality is defined through spatial locality. In the second subsection, we show that SLSQ model describes training of a particular class of RNNs. 

\subsection{Local SQ from local Gibbs training trajectory}
Under the noisy gradient descent training, we consider training with target state $\rho_{\text{ref}}$. The training gives a distribution $\mathbb{P}$ over trajectories $\{ \theta(0),...,\theta(T)\}$ of neural network parameters where each $\theta(t)$ is a vector of parameters. We make the assumption that $\rho_{\theta(t)}$ has a large probability of being Gibbs state of certain local Hamiltonians throughout the training process $t=0, \cdots T$. 

Before moving on, we comment on the intuition why we pose the ansatz that $\rho_{\theta(t)}$ is local Gibbsian. For distributions on a 1D lattice, the assumption is natural if both the initialization $\rho_{\theta(0)}$ and the target distribution $\rho_{\mathrm{ref}}$ have a finite Markov length. Recall Proposition 1 of Ref.~\cite{Aragon_s_Soria_2021}, which states that a distribution $\rho_{\theta}$ with decaying CMI is quasi-local Gibbsian, i.e., it can be well approximated by a Gibbs state of a local Hamiltonian. Therefore, both the initialization and the target are quasi-local Gibbsian. Furthermore, the neural networks (RNNs, CNNs and transformers) are empirically seen to prioritize learning local w and the local expectation values are sufficient to reconstruct the state $\rho_{\text{ref}}$ (the latter shown in App.~\ref{sec:AppE}). Thus, it is natural to think of the intermediate states $\rho_{\theta(t)}$ being Gibbs states of local Hamiltonians and interpolate between the two quasi-local Gibbs distributions. For higher dimensions, it is not firmly established if finite Markov length implies quasi-local Gibbsian distribution, so we will directly assume that $\rho_{\mathrm{ref}}$ is quasi-local Gibbsian. We also note that the notion of locality in this section, whether it be geometric locality or few-body locality, could depend on the architecture used. 

Now we consider the noisy gradient descent training targeting a state $\rho$ that is LI with $\rho_{\text{ref}}$. We will show that under the assumption of local Gibbs training trajectory above together with common regularity conditions, the training trajectories targeting $\rho$ and $\rho_{\mathrm{ref}}$ follow super-polynomially close distributions. Define $h(\theta) $ as the difference of the cross entropies:
\begin{equation}
    h(\theta) := \tr(\rho\log \rho_{\theta})- \tr(\rho_{\text{ref}} \log \rho_{\theta}).
\end{equation}
Since $\rho_{\theta(t)}$ has a large probability to be a Gibbs state of a local Hamiltonian of support on $l = O(1)$ sites, it follows from LI that
\begin{equation}
\label{eq:closeness}
    \mathbb{E}_{\theta(t) \sim \mathbb{P}} \left[ \left|h(\theta(t))\right| \right] < \delta
\end{equation}
Where $\delta \sim O(e^{-n^{1-c_0 }/l})$. Additionally, we will assume that the Hessians of the loss functions are bounded everywhere,
leading to 
\begin{equation}
    ||\nabla_{\theta_i} \nabla_{\theta_j} h(\theta) ||\leq L = \mathrm{poly}(n),
\end{equation}
where the norm $||\cdot||$ is the operator norm.
Finally, we will assume a regularity condition on the measure $\mathbb{P}$ that the expectation value of $h(\theta)$ under a super-polynomially small scalar perturbation $\epsilon$ to the parameter $\theta_i$ is still bounded, that is
\begin{equation}
\label{eq:closeness2}
    \mathbb{E}_{\theta \sim \mathbb{P}} [h(\theta(t) + \epsilon \mathbf{e}_i)] \leq \delta
\end{equation}
where $\mathbf{e}_i$ is a unit vector that perturbs $\theta_i$.

We proceed by expanding $h(\theta')$ around $\theta$ with respect to the component $\theta_i$. Let $\theta' = \theta + \epsilon \mathbf{e}_i$. There exists a point $\xi$ lying on the line segment connecting $\theta$ and $\theta'$ such that:
\begin{equation}
    h(\theta') = h(\theta) + \nabla_{\theta_i} h(\theta) \epsilon + \frac{1}{2} \epsilon^2 \nabla^2_{\theta_i} h(\xi)
\end{equation}
Rearranging the terms, we obtain:
\begin{equation}
    |\nabla_{\theta_i} h(\theta)| = \left|\frac{h(\theta') - h(\theta)}{\epsilon} - \frac{1}{2} \epsilon\nabla^2_{\theta_i} h(\xi)\right|
\end{equation}
Taking the expectation with respect to the distribution $\mathbb{P}$ over the trajectories on both sides yields:
\begin{equation}
    \mathbb{E}_{\theta \sim \mathbb{P}} \left[\left| \nabla_{\theta_i} h(\theta) \right| \right] = \mathbb{E}_{\theta \sim \mathbb{P}} \left[ \left|\frac{h(\theta') - h(\theta)}{\epsilon} - \frac{1}{2} \epsilon\nabla^2_{\theta_i} h(\xi)\right|\right]
\end{equation}
With the use of triangle inequality this becomes:
\begin{align}
     \mathbb{E}_{\theta \sim \mathbb{P}} \left[ \left|\nabla_{\theta_i} h(\theta) \right|\right]  &\leq \mathbb{E}_{\theta \sim \mathbb{P}} \left[ \left|\frac{h(\theta') - h(\theta)}{\epsilon}\right|\right] + \mathbb{E}_{\theta \sim \mathbb{P}}\left[ \left|\frac{1}{2} \epsilon\nabla^2_{\theta_i} h(\xi)\right|\right] \\ 
     &\leq \frac{\mathbb{E}_{\theta \sim \mathbb{P}} [\left| h(\theta')\right|]  + \mathbb{E}_{\theta \sim \mathbb{P}} [ \left|h(\theta)\right|] }{|\epsilon|} +  \mathbb{E}_{\theta \sim \mathbb{P}} \left[ \left|\frac{1}{2} \epsilon \nabla^2_{\theta_i} h(\xi) \right|\right] 
\end{align}
Using the assumptions Eqs.~\eqref{eq:closeness} - \eqref{eq:closeness2}, we arrive at:
\begin{equation}
     \mathbb{E}_{\theta \sim \mathbb{P}} \left[\left| \nabla_{\theta_i} h(\theta) \right|\right]  \leq \frac{2\delta}{|\epsilon|} + \frac{L}{2} |\epsilon|
\end{equation}
This holds for any $\epsilon$, and we can choose $\epsilon = \sqrt{4\delta/L}$, leading to the bound,
\begin{equation}
\label{eq:local_SQ_grad_bound}
    \mathbb{E}_{\theta \sim \mathbb{P}} \left[\left| \nabla_{\theta_i} h(\theta) \right|\right]  \leq 2\sqrt{L\delta}
\end{equation}
Thus, we have shown that, the difference in the gradients between training under LI target states $\rho$ and $\rho_{\text{ref}}$ remains super-polynomially small. This is the condition Eq.~\eqref{eq:LI_LSQDIFF} we need to prove hardness of learning LI states. Let $\mathbb{Q}$ be the distribution of the training trajectory $\{\theta(t)\}$ that targets the distribution $\rho$, the difference $||\mathbb{P}-\mathbb{Q}||_1$ is super-polynomially small, indicating hardness of learning LI states.

Finally, we comment on the relation to local SQ learning. First, recall that a noisy gradient descent training can be described by local SQ if $\nabla_{\theta}\log p_{\theta}$ is a local operator. Our core assumption that the state $\rho_{\theta(t)}$ is quasi-local Gibbsian does \text{not} imply $\nabla_{\theta}\log p_{\theta}$ is local without assuming further details on the architecture of the NN (see next subsection for a particular architecture where this does hold). Instead, the correct statement is that the training trajectory can be \textit{simulated} through local SQ oracles. To see this, we can explicitly restrict our ansatz to be local Gibbsian. Let us denote the new ansatz as $\sigma_{\tilde{\theta}}$. As we have shown, training trajectories targeting both $\rho$ is also local Gibbsian with a large probability, the training with $\sigma_{\tilde{\theta}}$ reproduces the same training trajectory up to negligible probability. Furthermore, $\nabla_{\tilde{\theta}}\log\sigma_{\tilde{\theta}}$ is explicitly a local operator. Therefore, training with this new ansatz is explicitly local SQ and simulates the original noisy gradient descent training.






\subsection{Spatially local SQ from Markovian RNN}
The state of an RNN is described by two functions $f_{\theta}$ and $g_{\theta}$, 
\begin{equation}
\begin{aligned}
    h_i &= f_{\theta}(h_{i-1}, s_{i-1}) \\
    p(s_i | h_i) &= g_{\theta}(s_i,h_i)
\end{aligned}
\end{equation}
where $\theta$ is the parameters of the RNN. The function $f_{\theta}$ determines the next hidden variable and the function $g$ determines the probability of output $s_i$ based on the hidden variable. The normalization $\sum_{s_i} p(s_i|h_i) = 1$ is usually enforced by a softmax function. The model is autoregressive, i.e., it gives a sample of each $s_i$ according to the distribution $p(s_i|h_i)$ and passes the result to the next cell. 

The dynamics of RNN is usually classified in three regimes: contractive, chaotic or edge of chaos based on the spectral norm of $\frac{\partial h_i}{\partial h_{i-1}}$. The contractive regime is defined as
\begin{equation}
    \left|\left| \frac{\partial h_i}{\partial h_{i-1}}\right|\right|< 1
\end{equation}
If the spectral norm is greater than 1, the RNN is known to be in the chaotic regime, leading to potential instability in training due to gradient blowing up. The edge-of-chaos regime has spectral norm close to 1, which has much more complicated dynamics. A convenient initialization in the contractive regime is to set the weights $\theta$ in $f_{\theta}$ to be small, which is adopted in NQS training. The sub-multiplicative property of the spectral norm indicates
\begin{equation}
   \left|\left| \frac{\partial h_{i+k}}{\partial h_i} \right|\right| \leq \prod_{j=i}^{i+k-1} \left|\left| \frac{\partial h_{j+1}}{\partial h_{j}}\right|\right|
\end{equation}
For long distances, the dependence of hidden variables is negligible. Therefore, an RNN in contractive regime can be described with a variable-length Markov machine - the distribution $p(s)$ has a finite Markov length. In this appendix, we further make an approximation: we choose $k$ sufficiently large such that the dependence can be neglected (e.g., due to finite numerical precision). 
\begin{equation}
\label{eq:markovrnn}
    \left|\left| \frac{\partial h_{i+k}}{\partial h_i} \right|\right| = 0
\end{equation}
This can be viewed as a coarse-graining of contractive RNNs, which we call Markovian RNNs. We show that training the RNN under this assumption is described by local SQ oracles.

Recall that in the unsupervised learning setup, the loss function $\mathcal{L}(\theta) = -\tr p_{\mathrm{target}}\log p_{\theta}$. At each training step, one computes the gradient $\nabla_{\theta}\mathcal{L}(\theta)$ by making an SQ oracle that returns the approximate expectation value of $\phi(s) = \nabla_{\theta} \log p_{\theta}(s)$. We show that $\phi(s)$ is a sum of local operators under the condition Eq.~\eqref{eq:markovrnn}. First, the full probability distribution is the product of the conditionals,
\begin{equation}
    p_{
    \theta
    }(s) = g_{\theta}(s_1, h_1) g_{\theta}(s_2, h_2) \cdots g_{\theta}(s_n,h_n).
\end{equation}
Using the chain rule, we have
\begin{equation}
    \nabla_{\theta} \log p_{\theta}(s) = \sum_{i=1}^n \frac{\partial}{\partial\theta}\log g_{\theta}(s_i, h_i) + \frac{\partial}{\partial h} \log g_{\theta}(s_i, h_i) \frac{d h_i}{d \theta}.
\end{equation}
where back propagation gives a recursive formula
\begin{equation}
\frac{d h_{i}}{d \theta} = \left( \frac{\partial h_{j}}{\partial h_{j-1}} \right)\left( \frac{d h_{j-1}}{d \theta} \right) +\frac{\partial f_{\theta}(h_{j-1}, s_{j-1})}{\partial \theta}.
\end{equation}
We can further iterate the back propagation $k$ steps, 
where the propagation stops due to Eq.~\eqref{eq:markovrnn}. Thus, $\frac{d h_i}{d \theta}$ only depends on $s_{i-k}$ to $s_{i-1}$. Thus, $\nabla_{\theta} \log p_{\theta}(s)$ is a sum of $k+1$-spatially local terms, justifying $s$-local SQ model.  
\section{Syndrome distribution of repetition code: MPDO representation}
Let us consider the syndrome distribution $q_p(s_1,s_2,\cdots s_n)$ of the repetition code with bit flip noise rate $p$. Each syndrome bit $s_i = (1 - b_i b_{i+1})/2 $, where $b_i = 1$ with probability $1-p$ and $b_i=-1$ with probability $p$ independently. We will use a MPDO to represent the classical state $\rho_p = \sum_{s} q_p(\vec{s})|\vec{s}\rangle \langle \vec{s}|$.

We show that the MPDO is a periodic uniform with bond dimension 2. The tensors are 
\begin{equation}
\begin{aligned}
    M_{00} &= (1-p) |0\rangle \langle 0| \\
    M_{01} &= p |1\rangle \langle 1| \\
    M_{10} &= (1-p) |1\rangle \langle
    1| \\
    M_{11} &= p |0\rangle\langle 0|
\end{aligned}
\end{equation}
Note that the indices are bond indices and the bra/ket is on physical indices. Since this represents a classical distribution, we don't have off-diagonal terms such as $|0\rangle \langle 1|$. Equivalently, we can write down the tensors by fixing the physical indices (syndromes) to be $|0\rangle\langle0|$ or $|1\rangle\langle1|$, 
\begin{equation}
M^{00} = \begin{bmatrix}
1 - p & 0 \\
0 & p
\end{bmatrix}
\quad \text{and} \quad
M^{11} = \begin{bmatrix}
0 & p \\
1 - p & 0
\end{bmatrix}
\end{equation}
Note that $M^{01} = M^{10} = 0$ again because the state is classical.  This gives the usual MPDO expression
\begin{equation}
\label{eq:MPDO}
    \rho_p = \mathrm{tr}(M^{s_1 s_1} \cdots M^{s_N s_N}) |s_1\cdots s_N\rangle\langle s_1 \cdots s_N|
\end{equation}

To see that the above MPDO tensors give the syndrome distribution $q_p$, we use the following celluar automata picture for the bond contraction.     
The celluar automata has two internal states: (1) State O: the site $i$ is not flipped, i.e., $b_i = 1$. (2) State I: the site $i$ is flipped, i.e., $b_i = -1$. The transition rules is that (1) O to O: we do not flip the site $i$ and the syndrome bit is $s_i = 0$. This gives $M_{00} = (1-p)|0\rangle\langle 0|$. (2) O to I: we flip the spin $i$ but do not flip spin $i-1$, resulting in syndrome bit $s_i = 1$. This gives $M_{01} = p|1\rangle\langle 1|$ (3) I to O: we do not the flip site $i$, giving a factor of $1-p$. Since the previous spin is flipped, the syndrome bit is $s_i = 1$. This gives $M_{10} = (1-p)|1\rangle\langle 1|$. (4) I to I: we flip the site $i$, giving a factor of $p$. Since the previous spin is also flipped, the syndrome bit is $s_i = 0$. This gives $M_{11} = p|0\rangle\langle 0|$. 


As a sanity check, if we choose $p=1/2$, then $M^{00} = I/2$ and $M^{11} = X/2$. As $X^2=I, \mathrm{tr}(I)=2, \mathrm{tr}(X)=0$, we are getting a global parity function
\begin{equation}
    \rho_{1/2} = \frac{1+\prod_{i=1}^N Z_i}{2^{N-1}}.
\end{equation}
The parity constraint, that is, $\prod_i s_i = 1$ is \textit{always} present for all $p$. This is because $M^{00}$ is diagonal and $M^{11}$ does not contain diagonal elements. To be more explicit, any matrix multiplication with an odd number of $M^{11}$ cannot contain diagonal elements, thus the trace vanishes.

Finally, we note that the MPDO can be transformed into an RNN of hidden dimension 4 using the map in Ref.~\cite{Wu_2023_MPSRNN}. Therefore, the syndrome distribution also admits an RNN representation.

\section{Syndrome distribution of repetition code: toy model for training}
Now let us consider a two-parameter ansatz $\rho_{p,r} = \mathcal{N}_r (\rho_p)$, where $\mathcal{N}_r$ is applying bit flip noise with probability $r$ on each \textit{syndrome bit}. This family of states represent a strong symmetry breaking perturbation to the target state $\rho_{p}$, while changing $p$ is a strong symmetry preserving perturbation. The density matrix $\rho_{p,r}$ is described by the following the MPDO ansatz with bond dimension 2,
\begin{equation}
\begin{aligned}
    M^{00}_{p,r} &= (1-r) M^{00} + rM^{11}  \\
    M^{11}_{p,r} &= r M^{00} + (1-r) M^{11}.
\end{aligned}
\end{equation}
Thus $\rho_{p,r}$ can also be represented with an RNN with hidden dimension 4. Instead of considering a full RNN training, let us consider the training this two-parameter ansatz $\rho_{p,r}$ by minimizing the loss function
\begin{equation}
    \mathcal{L}(p,r;p_0)=S(\rho_{p_0}||\rho_{p,r})
\end{equation}
The true minimum is achieved by $p^{*} = p_0,~ r^{*} = 0$.  At initialization of the training, we assume the state is initialized at $r\approx 1/2$, which corresponds to near maximally mixed state.
\subsection{Infinite Markov length $\xi = \infty$}
Let us first consider the fixed-point state, $p_0 = 1/2$ and $p=1/2$. We can compute the KL divergence analytically,
\begin{equation}
    \mathcal{L}(p=1/2,r;p_0=1/2)= \log \frac{2}{1+ (1-2r)^n}
\end{equation}
Consider the derivative with respect to $r$,
\begin{equation}
    \frac{\partial}{\partial r}\mathcal{L}(p=1/2,r;p_0=1/2) = \frac{2n(1-2r)^{n-1}}{1+ (1-2r)^n}
\end{equation}
Note that this gradient is exponentially small in $n$ if $r\neq 0$,
\begin{equation}
  \left|\frac{\partial}{\partial r}\mathcal{L}(p,r;p_0=p)  \right| = e^{-O(n)}, ~\mathrm{if~} p=\frac{1}{2}, r\neq 0
\end{equation}
Therefore, if the initialization is not exactly the target state, it cannot reach the target state in polynomial number of gradient step, even though the ansatz can represent the target state. Since $\mathcal{N}_r$ is a symmetry breaking perturbation, this shows that the presence strong symmetry in $\rho_{1/2}$ cannot be efficiently learned. The result agrees with the hardness of local SQ learning in App.~\ref{sec:AppE} and the hardness of general SQ learning in App.~\ref{sec:genSQ}.
\subsection{Finite Markov length}
Next, we consider the target state with $p = \frac{1}{2}-\delta p$, which has finite Markov length $\xi$. Note that the mixed-state RG analysis in Ref.~\cite{sang2024mixed} yields
\begin{equation}
    \xi \sim \left(\frac{1}{2}-p \right)^{-2}
\end{equation}
We again consider gradient of the loss function and show that the gradient $\partial \mathcal{L}/\partial r|_{r=1/2}$ is proportional to $n (\frac{1}{2}-p)^4$, which gives the scaling
\begin{equation}
\label{eq:divL0}
    \left|\frac{\partial}{\partial r}\mathcal{L}(p,r;p_0=p)|_{r=1/2}  \right| \propto \frac{n}{\xi^2} 
\end{equation}
To show this, we note that at $r=1/2$, the state $\rho_{p,r} = 2^{-n}I$ and therefore 
\begin{equation}
\label{eq:divL1}
     \left|\frac{\partial}{\partial r}\mathcal{L}(p,r;p_0=p)|_{r=1/2}  \right| = 2^n \sum_{s} \rho_{p}(s) \left.\frac{\partial \rho_{p,r}(s)}{\partial r}\right|_{r=1/2} 
\end{equation}
Let $r = \frac{1}{2}-\delta r$, we can expand the MPDO in the first order to obtain an MPDO tangent state, which gives
\begin{equation}
\label{eq:divL2}
    \rho_{p,r}(s) = 2^{-n} +  2^{-n+1} (1-2p)^2 (n- 2|s|)\delta r + O(\delta r^2)
\end{equation}
where $|s|$ is the Hamming weight of the bitstring $s$.
Combining Eqs.~\eqref{eq:divL1} and \eqref{eq:divL2} we obtain
\begin{equation}
    \left|\frac{\partial}{\partial r}\mathcal{L}(p,r;p_0=p)|_{r=1/2}  \right| = 8\delta p ^2 \sum_{s} \rho_p(s) (n-2|s|)
\end{equation}
The sum $\sum_{s}\rho_p(s)|s|$ is simply the average number of syndromes in the distribution $\rho_p(s)$. We can expand it to the second order in $\delta p$ as 
\begin{equation}
    \sum_{s}\rho_p(s)|s| = 2np(1-p) = \frac{n}{2}\left( 1-4\delta p^2 \right)
\end{equation}
Combining the two equations above we seed that the derivative is on the order of $O(n\delta p ^4)$, which gives Eq.~\eqref{eq:divL0}. The result implies that learning the parity constraint is polynomially hard in the Markov length $\xi$.

Note that we get a better training cost than the local SQ model (which is exponentially in $\xi$, see App.~\ref{sec:AppE}). This is because the number of variational parameters is much smaller, so we do not need full tomography of local marginal distributions. 

\section{SQ hardness of SWSSB states}
\label{sec:genSQ}
In this section we show that some of the hard-to-learn states we studied are not only hard in the local SQ model, but also hard to learn in the unrestricted SQ model where the query function is arbitrary. We first review the SQ oracles and SQ dimensions for decision problems over distributions. Then we generalize Kearns' proof of SQ hardness of parity functions to distributions. Note that the parity distribution $\rho_{1/2}$ is exactly the 0-form SWSSB fixed-point state. Finally, we consider a class of 1-form SWSSB states with zero correlation length, i.e., the loop ensemble with rigid constraint, and show that they are also hard under SQ. 
\subsection{Review of SQ learning of distributions}
The original SQ is formulated in terms of learning boolean functions \cite{kearns98,Aslam93SQ,Blum94sqdim,Blum03SQ,feldman12a,reyzin2020statisticalqueriesstatisticalalgorithms}. Here we are rather interested in its generalization \cite{feldman2016statisticalalgorithmslowerbound,diakonikolas2022optimalsqlowerbounds,diakonikolas2024statisticalquerylowerbounds} in terms of learning a distribution $p(x), x\in \mathbb{Z}^n_2$. We begin by reviewing the definition of the SQ oracle:
\begin{definition}[SQ oracle]
Given a distribution $p(x)$ of bitstrings, the learner has access to an SQ oracle $\mathsf{STAT}(\tau)$, where a query function $\phi(x): \mathbb{Z}^{n}_2\rightarrow [-1,1]$ is given. The oracle returns a value $v$ such that $v \in [\mathbb{E}_{x\sim p}(\phi(x))-\tau,\mathbb{E}_{x\sim p}(\phi(x))+\tau]$. 
\end{definition}
Our definition of $l$-local SQ oracle is a restriction of SQ oracle as it only allows for local query functions. We note that the SQ model describes almost all learning algorithms so far, including batch SGD in neural network states. In order to formulate learning hardness, we consider a class of decision problems.
\begin{definition} [Decision problems]
    We are given a distribution $q(x)$ and a set of distributions (a.k.a., a concept class)  $\mathcal{D} = \{p_i(x)\}$ that does not contain $q(x)$. We are given samples of either (1) $q(x)$ or (2) some $p_i(x)\in \mathcal{D}$. The problem is to decide whether (1) or (2) is true based on SQ oracles. 
\end{definition}
A decision problem is easy if the correct answer can be obtained with a probability at least $2/3$ with $T=\mathrm{poly}(n)$ SQ oracles and $\tau = 1/\mathrm{poly}(n)$. Otherwise the problem is hard. A key quantity to decide the hardness is the cross $\chi^2$-correlation matrix
\begin{equation}
    \chi_q(p_i, p_j) = \sum_{x} \left(\frac{p_i(x)}{q(x)}-1\right)\left(\frac{p_j(x)}{q(x)}-1\right)
\end{equation}
Note that if $i=j$, this reduces to the usual $\chi^2$-distance between the two distributions $p_i(x)$ and $q(x)$, which characterize their statistical distinguishability from samples. The correlation matrix determines the \textit{SQ dimension} $d$ that determines the hardness of SQ learning. 
\begin{definition} [SQ dimension]
    Given $\beta, \gamma>0$, the SQ dimension $d(\beta,\gamma)$ of a decision problem is given by the cardinality of the maximal subset $\mathcal{D}_S(\beta,\gamma)$ of $\mathcal{D}$, such that for any $p_i,p_j \in \mathcal{D}_S(\beta,\gamma)$, $|\chi_q (p_i, p_j)| \leq \gamma ~(i\neq j)$ and $|\chi_q (p_i, p_i)| \leq \beta$.
\end{definition}
The main result of Ref.~\cite{feldman2016statisticalalgorithmslowerbound} is an information-theoretic lower bound on the number $T$ of SQ oracles needed for a decision problem.
\begin{theorem}[Lemma 3.10 of Ref.~\cite{feldman2016statisticalalgorithmslowerbound}]
\label{thm:SQDIM}
    A decision problem is successful with probability $2/3$ only if we have $T\geq \frac{d\gamma}{\beta}$ oracles with tolerance $\tau = \sqrt{2\gamma}$, where $d = d(\beta,\gamma)$ is the SQ dimension.
\end{theorem}
Therefore, if a decision problem has an exponentially large SQ dimension for $\tau=1/\mathrm{poly}(n)$, then it is hard under SQ. Intuitively, an exponentially large SQ dimension means that a generic query function has an exponentially small signal because the distributions in the concept class $\mathcal{D}$ are almost orthogonal. 

\subsection{Hardness of parity distribution}

Now we prove that learning parity distribution
\begin{equation}
    p(x) = \frac{1+x_1 x_2\cdots x_n}{2^n}
\end{equation}
is hard. In order to show it, we use the doubling trick~\cite{abbe2022learningreasonneuralnetworks} to formulate a hard decision problem. Let $r(x) = 2^{-n}$ be the uniform distribution and  consider a $2n$-bit distribution 
\begin{equation}
    p_0(x,y) :=p(x)r(y).
\end{equation}
we construct a set $\mathcal{D} = \{p_S|S \in \mathbb{Z}^n_2\}$ such that
\begin{equation}
    p_S(x,y) = \left(\prod_{S_i = 1} \mathrm{SWAP}_{x_i,y_i} \right)p_0(x,y) 
\end{equation}
The operator in front of $p_0(x,y)$ means swapping all $x_i$'s into $y_i$'s if the $i$-th bit of $S$ is $1$. Note that $p_0\in \mathcal{D}$ as it corresponds to $S = 00\cdots 0$. The set of states essentially hides the $n$-bit parity constraint into of the total $2n$-bits.
Defining
\begin{equation}
    q(x,y) = r(x)r(y)
\end{equation}
to be the uniform distribution on $2n$ bits. We have a decision problem to tell whether a set of $2n$-bit samples comes from $q(x,y)$ or some distribution in $p_S(x,y)$. 

Now we show that the SQ dimension of this problem is exponentially large. We can compute 
\begin{equation}
    \chi_q(p_S,p_{S'}) = \delta_{SS'}
\end{equation}
because of the orthogonality of different boolean Fourier modes. Choosing $\beta = 1$ and $\gamma = 1/\mathrm{poly}(n)$, the whole set $\mathcal{D}$ satisfies that $|\chi_q (p_i, p_j)| \leq \gamma ~(i\neq j)$ and $|\chi_q (p_i, p_i)| \leq \beta$. Therefore, the SQ dimension is $d= |\mathcal{D}| = 2^{n}$. Then using theorem~\ref{thm:SQDIM} we need at least $T\geq 2^{n}/\mathrm{poly}(n)$ SQ oracles, establishing the hardness of the decision problem.

Finally, the hardness of the decision problem indicates the hardness of learning the distribution $p(x)$. One crucial assumption to make is that learning $p_S(x,y)$ is equally hard for all $S\in \mathbb{Z}^{n}_2$ with a random initialization of neural network, because they are related to $p_0$ by local reshuffling. Suppose learning $p(x)$ is easy, then learning $p_0(x,y)$ is also easy because we are just adding uniformly random bits. By the assumption, learning $p_S(x,y)$ is equally easy with random initializations. If the samples come from a distribution in $\mathcal{D}$, then this implies that we can use a polynomial amount of SQ oracles to determine that the sample indeed comes from a certain distribution in $\mathcal{D}$, solving the decision problem with certainty. However, the SQ hardness of the decision problem implies that this cannot be decided efficiently. This gives a contradiction and thus learning $p(x)$ is hard.

\subsection{Hardness of loop ensemble with rigid constraints}

\begin{figure}[t]
    \centering
    \includegraphics[width=\linewidth]{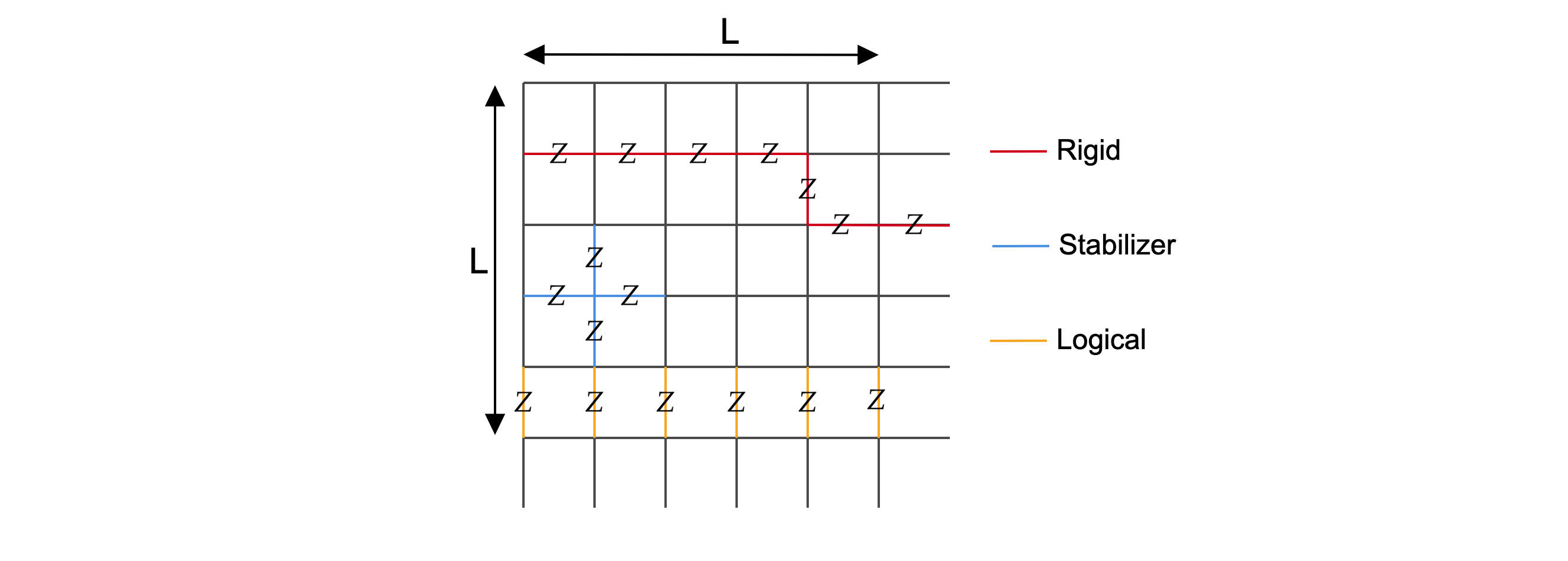}
    \caption{\textbf{Constraints in the loop model.} The logical constraints can be moved around by applying stabilizers. The rigid constraints are fixed on the lattice and cannot be moved by stabilizers.}
    \label{fig:LoopModelDiagram}
\end{figure}

Consider a 2D square lattice of sites $L\times  L$ with spin degrees of freedom living on the edges. For each vertex $v$, we can enforce the vertex constraint (a.k.a. stabilizers) $Z_v :=\prod_{i\in v} s_i = 1 $. The loop ensemble is the maximally mixed states with vertex constraint enforced.
\begin{equation}
    p_{\mathrm{loop}}(s)  = \frac{\prod_{v} (1 + Z_{v})}{2^{n}},
\end{equation}
where $n=2L^2$ is the number of spins. The distribution is an equal weight ensemble of closed loops, where the each $s_i=-1$ composes an edge to the loop. 

As noted in Ref.~\cite{sang2025mixedstatephaseslocalreversibility}, there are two types of constraints we could add to the $p_{\mathrm{loop}}(s)$ and the resulting state is still locally indistinguishable from $p_{\mathrm{loop}}(s)$. The first type is a logical constraint. There are only two logical operators $L_1$ and $L_2$. $L_1$ is the product of spins on the horizontal edges on the same vertical line. Analogously, $L_2$ is the product of spins on the vertical edges on the same horizontal line. The logical operator $L_1$ ($L_2$) can be equivalently defined on different vertical (horizontal) lines, as they are related by multiplying stabilizers.  Adding a logical constraint $L_1 = l_1 \in \{1,-1\}$ or $L_2 = l_2 \in \{1,-1\}$ results in a restricted ensemble
\begin{equation}
    p_{l_1 l_2} = p_{\mathrm{loop}}(s) (1 + l_1 L_1)(1 + l_2 L_2)
\end{equation}
Enforcing $L_1 = \pm 1$ constrains the loops in the ensemble to cross an even number of  
These four states are locally indistinguishable from $p_{\mathrm{loop}}$, thus they cannot be efficiently learned with \textit{local} SQ oracles.

The second type of constraint is rigid constraints. The rigid operator $R_i$ is defined as a product of spins on the edges that form a topologically nontrivial closed loop. Different loops correspond to different constraints as they are not related by multiplying stabilizers. We can define a loop ensemble with rigid constraint as
\begin{equation}
    p_i(s) = p_{\mathrm{loop}}(s) (1 + R_i)
\end{equation}
This constrains the loop configurations to cross an even number of times from the spins on $R_i$. Again, $p_i(s)$ is locally indistinguishable from $p_{\mathrm{loop}}(s)$.

Now we show that the decision problem with $q(s) = p_{\mathrm{loop}}(s)$ and $\mathcal{D} = \{p_i(s)\}$ is SQ hard. We can compute the $\chi^2$ correlation matrix
\begin{equation}
    \chi_q(p_i, p_j) =\langle R_i R_j\rangle_q
\end{equation}
where $R^2_i = 1$ and $R_i R_j$ is the product of spins on a topologically trivial loop if $i\neq j$. However, the since the vertex constraints are the only stabilizers in $q(s)$ and the loop constraints are independent from the vertex constraints,  $\langle R_i R_j\rangle_q = 0$ if $i\neq j$. Therefore, for the same reason as the parity functions, the SQ dimension $d = |\mathcal{D}|$, which is exponentially large in $L$, thus, the decision problem is SQ hard based on theorem~\ref{thm:SQDIM}.

\section{KL subtraction scheme for learning loop model under noise}
Consider the loop ensemble with fixed logical constraint $\rho_{l_1 l_2}$. Under the bit-flip noise, the state $\rho_{l_1l_2,p} = \mathcal{N}_p (\rho_{l_1 l_2})$ undergoes a decoding phase transition at $p_c\approx 0.109$. For $p<p_c$, states with different $l_1,l_2$ are locally indistinguishable. Therefore, local SQ predicts a learning hardness in for a fixed logical sector under $p<p_c$. In order to single out the hardness of learning the logical constraint from numerics, we have to take into account the hardness of learning the loop ensemble \textit{without} the logical constraints. Indeed, we find numerically that even if the target state is $\rho_{\mathrm{loop}}$ without logical constraints, we still get non-negligible $S(\rho_{\mathrm{loop}}||\rho_{\theta})$ after training the transformer or CNN state $\rho_{\theta}$. This leads to the KL subtraction scheme that we use in the main text, which we justify below.

Consider first the case of $p=0$. Let $\rho$ be the learned distribution for target $\rho_{\mathrm{loop}}$ at the final training step. We assume that $\rho$ does not contain any logical information because the target state is maximally mixed in the logical subspace.  Due to LI and the local SQ assumption, the same state $\rho$ is learned for the target states $\rho_{l_1l_2}$ with the logical constraint. The KL subtraction gives
\begin{equation}
\label{eq:KLdiff}
    \begin{aligned}
        S(\rho_{l_1 l_2}|| \rho) - S(\rho_{\mathrm{loop}}|| \rho) &= \left(\frac{1}{4} \sum_{l_1,l_2}   S(\rho_{l_1 l_2}|| \rho)\right) - S(\rho_{\mathrm{loop}}|| \rho) \\
        &= S(\rho_{\mathrm{loop} })-\left(\frac{1}{4} \sum_{l_1l_2} S(\rho_{l_1 l_2})\right)  - \tr \left[\left(\frac{1}{4} \sum_{l_1l_2} \rho_{l_1 l_2}\right) - \rho_{\mathrm{loop}}\right]\log \rho \\
        &=S(\rho_{\mathrm{loop}}) -\left(\frac{1}{4} \sum_{l_1l_2} S(\rho_{l_1 l_2})\right)  \\
        & =  S(\rho_{\mathrm{loop}}) - S(\rho_{l_1 l_2}) \\
        & = 2\log 2
    \end{aligned}
\end{equation}
where in the first line we have used that $\rho$ contain no logical information so the KL for different logical sectors are equal, in the third line we have used that $\rho_{\mathrm{loop}} = \frac{1}{4}\sum_{l_1 ,l_2} \rho_{l_1 l_2}$ and in the fourth line we use $S(\rho_{l_1 l_2})$ being independent of the logical $l_1$ and $l_2$. Next we consider the other fixed point $p=1/2$. The logical information is completely lost, i.e., $\mathcal{N}_p(\rho_{l_1 l_2}) =\mathcal{N}_p(\rho_{\mathrm{loop}})$, and thus the KL difference becomes to zero. Finally, we show that the KL difference has a phase transition at $p=p_c$. For $p<p_c$, the whole phase can be connected by two-way finite-depth classical channels:
\begin{equation}
\mathcal{D}(\rho_{l_1 l_2, p}) = \rho_{l_1 l_2}, ~~\mathcal{D}(\rho_{\mathrm{loop},p}) = \rho_{\mathrm{loop}}
\end{equation}
Let $\rho_p$ be the trained distribution for the target $\rho_{\mathrm{loop},p}$. Since LI is preserved under two-way channels, the target states $\rho_{l_1 l_2, p}$ are still LI and we get the same distribution $\rho_p$ under training. The same argument that leads to Eq.~\eqref{eq:KLdiff} gives
\begin{equation}
    S(\rho_{l_1 l_2,p}||\rho_p) - S(\rho_{\mathrm{loop},p}||\rho_p) = S(\rho_{\mathrm{loop},p}) - S(\rho_{l_1l_2, p})
\end{equation}
This entropy difference is exactly $2\log 2$ in the thermodynamic limit for all $p<p_c$. To see this, we note that the entropy difference is exactly the Holevo information: how much classical information about $l_1, l_2$ that we have in the states $\rho_{l_1 l_2,p}$. More precisely, the entropy difference is the mutual information of a state $\rho_{LS}$
\begin{equation}
    S(\rho_{\mathrm{loop},p}) - S(\rho_{l_1l_2, p}) = I(L: S)_{\rho_{LS}}
\end{equation}
with
\begin{equation}
    \rho_{LS} = \frac{1}{4}\sum_{l_1,l_2} |l_1 l_2\rangle^{L}\langle l_1 l_2| \otimes \rho^{S}_{l_1 l_2, p}
\end{equation}
For $p<p_c$, the mutual information remains $I(L:S) = 2\log 2$ because the logical information is decodable. Therefore, in this phase the KL difference is always $2\log 2$. For $p>p_c$, we have $I(L:S) = 0$ in the thermodynamic limit, which implies that $||\rho_{l_1 l_2, p} - \rho_{\mathrm{loop},p}||_1 \rightarrow 0$ in the thermodynamic limit. Therefore, the KL difference vanishes in the nondecodable phase.

The same analysis applies to the surface code under decoherence, where the KL difference is $\log 2 $ in the decodable phase because there is only one logical bit.

\section{Relation to simplicity bias}
We can supplement the difficulty of learning locally indistinguishable states with an intuitive explanation of why infinite markov length states are hard to learn. The dynamics of gradient descent have been both empirically and theoretically studied by ML community. Various neural network architectures trained with gradient descent exhibit a "simplicity" bias \cite{simplicityBias, spectralBias}. While the traditional notion of simplicity bias refers broadly to a preference for solutions with lower complexity, we consider a specific form as a spectral or degree bias. This mechanism has been touted as a potential explanation for the generalization capability of deep networks trained with gradient based methods. The simplicity bias we discuss is a specialization to the boolean case:  a preference for lower-degree terms in the Fourier-Walsh expansion \cite{sensitivefunctionshardtransformers}.

For a (pseudo) boolean function $f: \{-1, 1\}^N \rightarrow \mathbb{R}$, the (unique) fourier decomposition is defined as the expansion of $f$ in the orthormal basis of parity functions over subsets ($S \subseteq [N]$) of spins $\chi_S(x) = \prod_{i \in S} \sigma_{i}$, 
\begin{equation}
    f(x) = \sum_{S \subseteq [N]} \hat{f}(S) \chi_S(x)
\end{equation}
The simplicity bias of neural networks can be naively framed as the statement that a trained neural network exhibits the following property: $\hat{f_\theta}(S) \rightarrow 0$ as $|S| \rightarrow N$. The property of gradient descent finding such "simple" solutions, and the general difficulties of neural networks to learn dense parities have been extensively studied for other learning paradigms \cite{abbe2022learningreasonneuralnetworks, abbe2025learninghighdegreeparitiescrucial}. 
Within the context of unsupervised distribution learning, boolean expansion of both $p(x)$ or $\log p(x)$ have nice physical interpretations. In the case of $p(x)$, the fourier coefficients are proportional to the spin moments, $\hat{p}(S) \propto \expval{\chi_S}$. For $\log p$, the boolean expansion essentially amounts to writing an effective hamiltonian for which $p$ is a gibbs state.
\begin{equation}
    \log{p(x)} = H_{\text{eff}}=  \sum_S \hat{c}(S) \chi_S(x)
\end{equation}

We claim that the simplicity bias in the context of unsupervised learning can be framed in terms of these effective hamiltonian coefficients. We hypothesize that autoregressive neural networks trained in this paradigm are biased towards gibbs states of hamiltonians with low degree terms. The architectural inductive biases determine the precise form this takes. In architectures where spatial locality is important such as shallow CNNs or recurrent neural networks, the bias is towards gibbs states of spatially local low-degree hamiltonians. Other architectures such as Transformers or deep CNNs may not have this spatial locality constraint and are biased towards gibbs states of low degree (``k-local'') hamiltonians without requiring spatial proximity of interacting spins. 


\bibliographystyle{apsrev4-2}
\bibliography{refs}
\end{document}